\documentclass[twocolumn,10pt]{IEEEtran}

\usepackage{soul, xcolor}
\usepackage{booktabs}
\usepackage{cite}
\usepackage{graphicx}
\usepackage{psfrag}
\usepackage{url}
\usepackage{amsmath}
\usepackage{array}
\usepackage{amssymb}
\usepackage{mathtools}
\usepackage{amsfonts}
\usepackage{graphicx}
\usepackage{epstopdf}
\usepackage{algorithm}
\usepackage{cite,algorithm,algorithmic,amsmath,amssymb,amsthm,empheq,mhsetup}
\usepackage{algorithmic}
\newtheorem{proposition}{Proposition}

\newtheorem{remark}{Remark}

\newtheorem{theorem}{Theorem}
\newtheorem{assumption}{Assumption}
\newtheorem{conjecture}{Conjecture}

\usepackage{setspace}
\usepackage{amscd}
\usepackage{mathrsfs}
\usepackage{epsfig}
\usepackage{color}
\usepackage{textcomp}
\usepackage{float}
\usepackage{pgf}
\usepackage{pgfplots}
\usepackage{tikz}
\usetikzlibrary{spy}
\usepackage{eqparbox}
\usepgfplotslibrary{groupplots}
\usetikzlibrary{decorations.pathreplacing,calligraphy}
\usepackage{lipsum}  
\usepackage{nicematrix}
\usepackage{glossaries}
\usepackage{glossaries-extra}
\usepackage{gensymb}
\usepackage{soul}

\IEEEoverridecommandlockouts

\DeclareMathOperator{\trace}{Tr}

\DeclareMathOperator*{\argmin}{argmin}

\setabbreviationstyle[acronym]{long-short}
\newacronym{fwa}{FWA}{fixed wireless access}
\newacronym{mimo}{MIMO}{multiple input multiple output}
\newacronym{5g}{5G}{fifth generation}
\newacronym{iid}{i.i.d}{independent and identically distributed}
\newacronym{bs}{BS}{base station}
\newacronym{ue}{UE}{user equipment}
\newacronym{ap}{AP}{access point}
\newacronym{cpu}{CPU}{central processing unit}
\newacronym{upa}{UPA}{uniform planar array}
\newacronym{los}{LoS}{line-of-sight}
\newacronym{awgn}{AWGN}{additive white Gaussian noise}
\newacronym{isp}{ISP}{internet service provider}
\newacronym{mmse}{MMSE}{minimum mean square error}
\newacronym{mse}{MSE}{mean square error}
\newacronym{zf}{ZF}{zero-forcing}
\newacronym{mrc}{MRC}{maximum ratio combiner}
\newacronym{mrt}{MRT}{maximum ratio transmission}
\newacronym{se}{SE}{spectral efficiency}
\newacronym{csi}{CSI}{channel state information}
\newacronym{snr}{SNR}{signal-to-noise ratio}
\newacronym{sinr}{SINR}{signal-to-interference-plus-noise ratio}

\DeclareMathOperator{\rank}{\mathrm{rank}}
\DeclareMathOperator{\EEE}{\mathbb{E}}
\DeclareMathOperator{\var}{\mathrm{Var}}

\DeclareMathOperator{\C}{\mathbb{C}}
\DeclareMathOperator{\F}{\mathbf{F}} 
\DeclareMathOperator{\f}{\mathbf{f}}
\DeclareMathOperator{\aaa}{\mathbf{a}}

\DeclareMathOperator{\SSSS}{\mathbf{S}}

\DeclareMathOperator{\OO}{\mathcal{O}}

\DeclareMathOperator{\E}{\mathbf{E}}
\DeclareMathOperator{\V}{\mathbf{V}}
\DeclareMathOperator{\vv}{\mathbf{v}}
\DeclareMathOperator{\w}{\mathbf{w}}
\DeclareMathOperator{\z}{\mathbf{z}}

\DeclareMathOperator{\h}{\mathbf{h}}
\DeclareMathOperator{\HH}{\mathbf{H}}

\DeclareMathOperator{\R}{\mathcal{R}}

\DeclareMathOperator{\RR}{\mathbf{R}}

\DeclareMathOperator{\G}{\mathbf{G}}
\DeclareMathOperator{\D}{\mathbf{D}}

\DeclareMathOperator{\TT}{\mathcal{T}}

\DeclareMathOperator{\CN}{\mathcal{CN}}
\DeclareMathOperator{\A}{\mathbf{A}}

\DeclareMathOperator{\B}{\mathbf{B}}
\DeclareMathOperator{\bb}{\mathbf{b}}

\DeclareMathOperator{\N}{\mathbf{N}}
\DeclareMathOperator{\NN}{\mathcal{N}}

\DeclareMathOperator{\CC}{\mathbf{C}}

\DeclareMathOperator{\W}{\mathbf{W}}

\DeclareMathOperator{\PP}{\mathbf{P}}
\DeclareMathOperator{\p}{\mathbf{p}}
\DeclareMathOperator{\II}{\mathbf{I}}

\DeclareMathOperator{\x}{\mathbf{x}}

\DeclareMathOperator{\y}{\mathbf{y}}

\DeclareMathOperator{\Q}{\mathbf{Q}}

\DeclareMathOperator{\g}{\mathbf{g}}

\DeclareMathOperator{\X}{\mathbf{X}}
\DeclareMathOperator{\Y}{\mathbf{Y}}
\DeclareMathOperator{\Z}{\mathbf{Z}}
\DeclareMathOperator{\0}{\mathbf{0}}

\DeclareMathOperator{\GAMMA}{\boldsymbol{\Gamma}}

\DeclareMathOperator{\LAMBDA}{\boldsymbol{\Lambda}}

\DeclareMathOperator{\DELTA}{\boldsymbol{\Delta}}

\DeclareMathOperator{\vecc}{\mathrm{vec}}
\DeclareMathOperator{\Cov}{\mathrm{Cov}}
\DeclareMathOperator{\Bias}{\mathrm{Bias}}
\DeclareMathOperator{\MSE}{\mathrm{MSE}}

\makeatletter
\newcommand\fs@spaceruled{\def\@fs@cfont{\bfseries}\let\@fs@capt\floatc@ruled
  \def\@fs@pre{\vspace{0.5\baselineskip}\hrule height.7pt depth0pt \kern2pt}%
  \def\@fs@post{\kern2pt\hrule\relax}%
  \def\@fs@mid{\kern2pt\hrule\kern2pt}%
  \let\@fs@iftopcapt\iftrue}
  
\newsavebox\myboxA
\newsavebox\myboxB
\newlength\mylenA
\newcommand*\mybar[2][0.66]{
	\sbox{\myboxA}{$\m@th#2$}
	\setbox\myboxB\null
	\ht\myboxB=\ht\myboxA
	\dp\myboxB=\dp\myboxA
	\wd\myboxB=#1\wd\myboxA
	\sbox\myboxB{$\m@th\overline{\copy\myboxB}$}
	\setlength\mylenA{\the\wd\myboxA}
	\addtolength\mylenA{-\the\wd\myboxB}
	\ifdim\wd\myboxB<\wd\myboxA
	\rlap{\hskip 0.5\mylenA\usebox\myboxB}{\usebox\myboxA}
	\else
	\hskip -0.5\mylenA\rlap{\usebox\myboxA}{\hskip 0.5\mylenA\usebox\myboxB}
	\fi}
  
\makeatother

\usepackage{tkz-euclide}
\tikzset{
	bs/.pic = {                                      
		\draw[line width = 1pt,-round cap] (0,0.4\R) -- (-0.2\R,-0.2\R);
		\draw[line width = 1pt,-round cap] (0,0.4\R) -- (0.2\R,-0.2\R);
		\draw[line width = 1pt] (-0.2\R,-0.2\R) -- (0.133\R,0);
		\draw[line width = 1pt] (0.2\R,-0.2\R) -- (-0.133\R,0);
		\draw[line width = 1pt,-round cap] (-0.133\R,0) -- (0.067\R,0.2\R);
		\draw[line width = 1pt,-round cap] (0.133\R,0) -- (-0.067\R,0.2\R);
		\draw[line width = 1pt,-round cap] (-0.213\R,0.4\R) -- (0.213\R,0.4\R);
		\draw[line width = 1pt,-round cap] \foreach \x in {-0.21, -0.14,...,0.22} {(\x\R,0.4\R) -- (\x\R,0.47\R)};
		\node (dim) at (0,0)  [align=center,minimum width=0.5\R,minimum height=1\R] {};
	},
}

\tikzset{
	user/.pic = {     
		\draw[rounded corners=0.02\R] (-0.09\R,-0.17\R) rectangle (0.09\R,0.17\R) ;                     
		\draw[fill=gray] (-0.08\R,-0.12\R) rectangle (0.08\R,0.12\R);                          
		\draw[rounded corners=0.005\R] (-0.04\R,0.14\R) rectangle (0.04\R,0.15\R);                     
		\draw (0,-0.14\R) circle (0.012\R) ; 
		\node (dim) at (0,0)  [align=center,minimum width=0.2\R,minimum height=0.3\R] {};
	}
}

\usetikzlibrary{external}
\usepgfplotslibrary{external}

\allowdisplaybreaks
\begin{document}

\title{Leveraging Slowly Time-Varying AP-AP Channels for Interference Mitigation in Dynamic TDD}

\author{Martin Andersson,~\IEEEmembership{Graduate Student Member,~IEEE,}
    Tung T. Vu,~\IEEEmembership{Member,~IEEE,}
    \\
    Pål Frenger,
    Jan Åslund,
    and Erik G. Larsson,~\IEEEmembership{Fellow,~IEEE}
    
    \thanks{Parts of this paper were presented at the 2024 International Conference on Acoustics, Speech, and
    Signal Processing (ICASSP), see \cite{Andersson2024ICASSP}.}
    
    \thanks{M. Andersson, J. Åslund and E. G. Larsson are with the Department of Electrical Engineering (ISY), Link\"{o}ping University, 581 83 Link\"{o}ping, Sweden. Email: \{martin.b.andersson, jan.aslund, erik.g.larsson\}@liu.se.}
    
    \thanks{T. T. Vu is with the School of Computer, Data and Mathematical Sciences, Western Sydney University, Australia. Email: t.vu4@westernsydney.edu.au.}

    \thanks{P. Frenger is with Ericsson Research, 583 30 Link\"{o}ping, Sweden. Email: pal.frenger@ericsson.com.}
    
    \thanks{This work was partially supported by the Wallenberg AI, Autonomous Systems and Software Program (WASP) funded by the Knut and Alice Wallenberg Foundation, and partially supported by ELLIIT. The work of T. T. Vu was partially supported by the Australian Research Council Discovery Early Career Researcher Award (DECRA) under Project DE260100816.}
}

\maketitle
\thispagestyle{empty}

\begin{abstract}
    We address the challenge of cross-link interference in dynamic time-division duplexing (TDD) systems. Specifically, we focus on mitigating the interference caused by access points (APs) operating in downlink to APs operating in uplink. To this end, we exploit that channels between APs typically vary much more slowly over time than channels between users and APs. This observation allows us to jointly estimate the uplink user data and the AP-AP channels using a least-squares formulation over multiple coherence intervals, during which the AP-AP channels stay constant. We derive conditions for unique solvability of this least-squares problem by analyzing the rank of the regression matrix. For cases where a unique solution does not exist, we propose to transform the problem into a uniquely solvable one by sacrificing a subset of the uplink data samples. Numerical results demonstrate that our proposed methods achieve substantial gains over baseline algorithms. Further, we observe that one of our proposed algorithms achieves almost perfect AP-AP interference mitigation when the AP-AP channels vary very slowly over time.
\end{abstract}
\vspace{-3mm}

\glsresetall

\vspace{-6mm}
\section{Introduction}
\label{sec:intro}

Modern multiple-input multiple-output (MIMO) technologies, including massive (cellular) MIMO and distributed (cell-free) MIMO, are designed to operate in time-division duplexing (TDD) with half-duplex access points (APs) \cite{ngo16,emil20TWC,ammar22CST}. Conventionally, these systems operate in \emph{static} TDD, where the fraction of time allocated to uplink (UL) respectively downlink (DL) transmissions is fixed and aligned among all APs. Contrarily, in \emph{dynamic} TDD operation, each individual AP can adjust its UL and DL allocation in every coherence interval. This allows dynamic TDD systems to cater asymmetric and time-varying UL and DL traffic by adjusting the allocation based on the instantaneous data demands \cite{Andersson2023Asilomar}. However, dynamic TDD introduces cross-link interference among APs and among user equipment (UEs), and this interference must be managed to maximize the benefits of dynamic TDD \cite{KimCST}.

Several studies have shown advantages of dynamic TDD over static TDD, predominantly in the special case where each AP is restricted to operate in either UL or DL throughout the entire frame, i.e., the APs do not switch between UL and DL within a frame \cite{Mohammadi2023JSAC,Chowdhury2024TCOM}. This operation mode is often referred to as network-assisted full-duplex and will be considered in this work. Further, dynamic TDD has been discussed and is featured in the 3GPP standards \cite{5GNR_crosslink, 3GPP_config}, which shows its potential for implementation in practical systems.

The use of half-duplex APs makes dynamic TDD a cost-effective flexible duplexing mode. It can be enabled in existing MIMO systems, in contrast to full-duplex schemes that require additional complex and costly hardware for self-interference cancellation \cite{Zhang2015CM, Smida2024PI,Kim2024PI,Mohammadi2025TCOM}. In fact, networks with half-duplex APs and dynamic TDD operation can outperform corresponding full-duplex networks, both in terms of spectral and energy efficiency \cite{chowdhury2021TC, Chowdhury2024TCOM, Mohammadi2023JSAC}. For full-duplex to compete with dynamic TDD, the self-interference within the full-duplex APs must be highly mitigated, enforcing strict hardware requirements.

Although dynamic TDD systems do not suffer from self-interference, the cross-link interference remains. Specifically, the AP-AP interference (AAI), which occurs when neighboring APs operate in opposite transmission directions on the same time-frequency resources, can significantly limit the UL data rates in dynamic TDD unless mitigated \cite{Andersson2024ICASSP, KimCST}.

\subsection{State-of-the-Art in AP-AP Interference Mitigation}

Existing methods for AAI mitigation primarily rely on strict coordination of the AP scheduling, beamforming, or power control. One exception is \cite{Huang2019SJ}, wherein the authors show that the AAI can be perfectly mitigated in dynamic TDD massive MIMO systems when the number of AP antennas per UE tends to infinity. However, several hundreds of antennas per AP are required to get close to this limit, 
which is not always feasible in practical systems.

More specifically, the papers \cite{Zhu2021CL, Razlighi2021TWC,fukue22A,Mohammadi2023JSAC,Chowdhury2024TCOM} optimize the system throughput by scheduling the transmission modes (UL or DL) of the APs in each coherence interval. In \cite{Zhu2021CL}, the AP mode selection is optimized in order to maximize the sum UL/DL SE, whereas \cite{Mohammadi2023JSAC} and \cite{Chowdhury2024TCOM} formulate similar problems that also include AP power control in the optimization problem. Further, \cite{fukue22A} maximizes the geometric mean of the per-UE spectral efficiency to achieve fairness, including the DL beamforming design in the optimization problem. 

The existing works closest to ours consider coordinated AP-AP beamforming methods where the transmit precoding at the DL APs or the receive combining at the UL APs is designed to suppress AAI \cite{Yoon2015CL,deOlivindo2018WCL,Jayasinghe2018TSP,daSilva2021WC}. For example, \cite{daSilva2021WC} applies zero-forcing receive combining at the UL APs to nullify the AAI from the DL APs, while \cite{deOlivindo2018WCL} designs the beamforming at the DL APs such that the AAI at the UL APs is constrained. 

An important limitation of the existing methods for AAI mitigation, specifically those based on coordinated beamforming cited above, is their reliance on precise AP-AP channel knowledge; they require frequent estimation of these channels. This introduces a vast overhead for transmission of pilot sequences between APs for AP-AP channel estimation. During this pilot phase, the DL APs must halt their data transmission (reducing the DL pre-log factor) and the UL UEs must be silent to avoid corrupting the AP-AP channel estimates at the UL APs (reducing the UL pre-log factor), which heavily limits performance and scalability. Furthermore, most of these methods require adaptation of the DL transmissions (e.g., through DL beamforming design) in order to alleviate the AAI, wasting degrees of freedom on interference mitigation and impairing the DL data rates. 

\vspace{-3mm}
\subsection{Contributions of the Paper}

In this paper, we develop an alternative strategy by exploiting that the UL APs know the data signals transmitted by the DL APs. This lets us use the DL transmit signals as implicit pilots to avoid an overhead for AP-AP channel estimation. However, these DL transmit signals are rank-deficient, and hence cannot be used for AP-AP channel estimation within a single coherence interval. But provided that the AP-AP channels are slowly time-varying (compared to the UE-AP channels)\footnote{AP-AP channels varying more slowly over time than UE-AP channels is a robust physical attribute. This is because APs are static in space whereas UEs are mobile, implying faster variations in the UE-AP channels.}, the same AP-AP channel realization will experience multiple independent DL signals, which allows us to observe all dimensions of the channel matrix over time.

More specifically, we formulate the AAI mitigation problem as a joint estimation of the UL data and the AP-AP channel over multiple coherence intervals during which the AP-AP channel is constant. 
Mathematically, this gives rise to a coupled (two-sided) least-squares problem of the form
\begin{align}
    \label{eq:LS}
    \min_{ \{\X_t\}, \HH } \ \sum_{t \in \TT} \| \Y_t - \G_t \X_t - \HH \V_t \SSSS_t \|^2,
\end{align}
where $\{ \X_t \}$ are the UL data matrices, and $\HH$ is the  AP-AP channel. The matrices $\{ \Y_t \}$, $\{ \G_t \}$, $\{ \V_t \}$, and $\{ \SSSS_t \}$ are known. (All notation will be defined in detail later.)

Solving \eqref{eq:LS} is straightforward assuming that a unique solution exists. However, to the best of our knowledge, the existence of a unique solution to problems of the form \eqref{eq:LS} has not been analyzed in the literature, not even from a purely mathematical perspective. Some results exist on the (unique) solvability of matrix equations of the form $ \Y = \G \X + \HH \V \SSSS$ with respect to $\X$ and $\HH$ (which corresponds to our scenario with $| \TT | = 1$ in the noise-free case) \cite[Ch. 4]{horn1991topics}, \cite{Roth1952,Baksalary1979LAA}. However, these results do not generalize to equation systems that define the solution of \eqref{eq:LS} (with a summation over $t$). Other studies have analyzed general systems of coupled matrix equations, whereof our scenario is a special case. For example, \cite{Ding2006SIAM} develops iterative algorithms to solve systems of coupled matrix equations, but it assumes that a unique solution of \eqref{eq:LS} exists and does not characterize detailed conditions for the existence of a unique solution. Problem \eqref{eq:LS} can also be expressed on vector form using the Khatri-Rao product \cite{khatri1968}, and one possible approach could be to bound the rank of the resulting effective regression matrix, using results from multilinear decomposition theory. Specifically,
\cite{Sidiropoulos2000JoC} and \cite{DeLathauwer2008SIAM} give $k$-rank and $k'$-rank bounds for Khatri-Rao products, which could in principle be used to bound the rank of our regression matrix. But these rank bounds are too weak to be useful in our context. 

Our technical contributions are:
\begin{enumerate}
    \item We formulate the AAI mitigation problem as in \eqref{eq:LS}, by leveraging the slow time-variation of the AP-AP channels and knowledge of the DL transmit signals. Further, we calculate the mean-square error (MSE) and derive an achievable spectral efficiency (SE) for the UL data symbol estimates given by the solution to \eqref{eq:LS}. Numerical simulations show that our proposed method enables effective AAI mitigation and significantly outperforms baseline receivers both in terms of MSE and SE. When the AP-AP channel varies very slowly over time, the performance of our proposed technique approaches that of a genie-aided baseline with perfect AAI mitigation.

    \item For cases where \eqref{eq:LS} does not have a unique solution, we propose a modified (uniquely solvable) least-squares problem by silencing a subset of the UL data symbols in each coherence interval.
    
    \item As a methodological contribution to least-squares estimation, we investigate the existence of a unique solution to \eqref{eq:LS} by analyzing the rank of its regression matrix.  We provide the following specific results (we refer to Section \ref{sec:analysis} for the formal statements):
    \begin{itemize}
        \item \textbf{Theorem \ref{theorem:dimNA}:} A \emph{necessary} condition in terms of $| \TT |$ and the dimensions of the matrices in \eqref{eq:LS}, for which the least-squares problem has a unique solution.
        \item \textbf{Theorem \ref{theorem:M=N=2K=2J}:} The necessary condition for identifiability in Theorem \ref{theorem:dimNA} is \emph{sufficient} in the case where there are equally many UL and DL AP antennas, and the number of UL/DL UEs present in the network is half of the number of AP antennas.
        \item \textbf{Theorem \ref{theorem:T=1}:} In the case $| \TT | = 1$,  the least-squares problem \eqref{eq:LS} is \textit{never} identifiable.
    \end{itemize}
    
\end{enumerate}

This paper is a comprehensive extension of our conference paper  \cite{Andersson2024ICASSP},
which was limited to single-slot processing (i.e, $| \TT |=1$ above) and to the case when $\V_t\SSSS_t$ has rank one.

\subsection{Notation}

We denote scalars by regular lowercase symbols. Bold symbols in lowercase and uppercase are column vectors and matrices, respectively. The set of complex numbers is denoted by $\C$; the sets of complex vectors and matrices include a superscript of their dimension. We use $\left( \cdot \right)^*$, $\left( \cdot \right)^T$, $\left( \cdot \right)^H$ for complex conjugate, transpose, and Hermitian transpose. The Kronecker product is denoted by $\otimes$ and the vectorization operator by $\vecc\left( \cdot \right)$. For square matrices, we use $\trace\left(\cdot\right)$, $\det\left(\cdot\right)$, and $\left( \cdot \right)^{-1}$ for the trace, determinant, and matrix inverse. The nullspace of a matrix is denoted by $\NN\left( \cdot \right)$ and the dimension of a vector space by $\dim\left( \cdot \right)$. The covariance matrix, bias, and mean square error of a random vector are, respectively, denoted by $\Cov\left( \cdot \right)$, $\Bias\left( \cdot \right)$, and $\MSE\left( \cdot \right)$. We use $\CN\left( \cdot, \cdot \right)$ to denote a complex Normal distribution. The notation $\| \cdot \|$ represents the Euclidean norm of a vector or the Frobenius norm of a matrix. For matrices, the notation $\left[ \cdot \right]_{i:j, k:l}$ is the submatrix obtained by taking the intersection of rows $i,\dots,j$ and columns $k,\dots,l$.

\begin{figure*}[h!]
    \centering
    \begin{tikzpicture}
        \newdimen\R
        \R=2cm
        \draw (1,0) pic (bsA) {bs} node [xshift=-6mm, yshift=8mm]{};
        \draw (7,0) pic (bsB) {bs} node [xshift=6mm, yshift=8mm] {};
        \draw (6,-1.7) pic (dlue1) {user} node [yshift=-6mm]{$1$} ;
        \draw (6.7,-1.8) pic (dlue2) {user} node [yshift=-6mm]{$j$} ;
        \draw (8,-1.7) pic (dlue3) {user}  node [yshift=-6mm]{$J$} ;
        \draw (0,-1.7) pic (ulue1) {user} node [yshift=-6mm]{$1$} ;
        \draw (0.7,-1.8) pic (ulue2) {user} node [yshift=-6mm]{$k$} ;
        \draw (2,-1.7) pic (ulue3) {user}  node [yshift=-6mm]{$K$} ;
        \draw[<-,thick,shorten >= 0.2cm] (0.8,-0.5) to node [align=center,midway,xshift=4mm,yshift=0mm]{} (ulue1dim.north) ;
        \draw[<-,thick,shorten >= 0.2cm] (1,-0.5) to node [align=center,midway,xshift=4mm,yshift=0mm]{{ }} (ulue2dim.north) ;
        \draw[<-,thick,shorten >= 0.2cm] (1.2,-0.5) to node [align=center,midway,xshift=4mm,yshift=0mm]{} (ulue3dim.north);
        \draw[->,thick,shorten >= 0.2cm] (6.8,-0.5) to node [align=center,midway,xshift=4mm,yshift=0mm]{} (dlue1dim.north) ;
        \draw[->,thick,shorten >= 0.2cm] (7,-0.5) to node [align=center,midway,xshift=4mm,yshift=0mm]{{ }} (dlue2dim.north) ;
        \draw[->,thick,shorten >= 0.2cm] (7.2,-0.5) to node [align=center,midway,xshift=4mm,yshift=0mm]{} (dlue3dim.north);
        
        \draw[thick, blue] (1,-0.76) ellipse (0.8cm and 0.35cm) node [align=center,xshift=20mm,yshift=0mm]{Data channels: \\ $\G_t \in \C^{M \times K}$};
        
        \draw[red,thick,<-] (bsAdim.north) to[out=45,in=135] node [align=center,midway,xshift=0mm,yshift=-7mm]{Interference channel: \\ $ \HH \in \C^{M \times N}$} (bsBdim.north) ;

        \node  [align=left,left=5mm of bsAdim.north] {UL AP: \\ $\bullet$ $M$ antennas};
        \node  [align=left,right=5mm of bsBdim.north] {DL AP: \\ $\bullet$ $N$ antennas \\ $\bullet$ Data $\SSSS_t \in \C^{J \times \tau}$ \\ $\bullet$ Precoder $\V_t \in \C^{N \times J}$};
        \node  [align=left,left=5mm of ulue1dim.north] {UL UEs: \\ $\bullet$ Data $\X_t \in \C^{K \times \tau}$};
        \node  [align=center,right=5mm of dlue3dim.north] {DL UEs};         
        \node at (1.2,-1.78)[circle,fill,inner sep=0.8pt]{};
        \node at (1.35,-1.76)[circle,fill,inner sep=0.8pt]{};
        \node at (1.5,-1.74)[circle,fill,inner sep=0.8pt]{};
        
        \node at (7.2,-1.78)[circle,fill,inner sep=0.8pt]{};
        \node at (7.35,-1.76)[circle,fill,inner sep=0.8pt]{};
        \node at (7.5,-1.74)[circle,fill,inner sep=0.8pt]{};
    \end{tikzpicture}
    \caption{Illustration of the considered MIMO system operating in dynamic TDD with AP-AP interference.}
    \label{Fig:system}
\end{figure*}
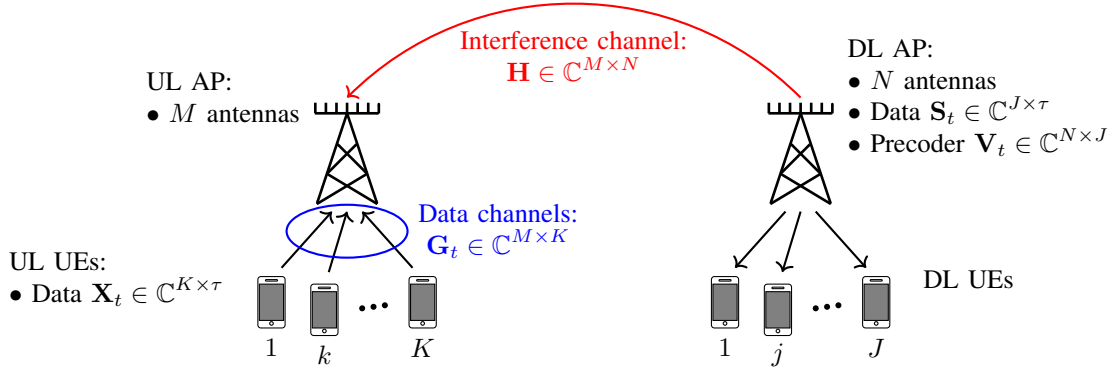

\section{System Model}
\label{sec:model}

We consider the UL data transmission in a MIMO system with dynamic TDD operation. Specifically, one AP equipped with $M$ antennas receives the UL data transmitted from $K$ single-antenna UEs. The reception of the UL data is interfered by the cross-link interference caused by another AP with $N$ antennas that transmits DL data to $J$ single-antenna UEs, on the same time-frequency resources.\footnote{We focus on the special case of dynamic TDD operation where UL and DL transmissions take place simultaneously in the entire duration of each coherence interval. In practice, the two APs might operate in separate directions only during a fraction of each coherence interval. Our theoretical framework is still applicable within that fraction. Importantly, our methods are agnostic to how the scheduling has been performed. We take a given UL/DL AP assignment and improve the performance by adding AP-AP interference mitigation on top of a given dynamic TDD schedule.
} Hereafter, we will refer to the AP/UEs operating in UL and DL as the UL and DL AP/UEs, respectively. Further, we will assume that $K < M$ and $J < N$. Our system model is illustrated in Fig. \ref{Fig:system}. 
\begin{remark}
    For clarity of presentation, we consider the case with co-located AP antennas and single-antenna UEs. However, our system model is general and covers the cases of arbitrarily distributed AP antennas and multi-antenna UEs.
\end{remark}
In the sequel, we refer to the cross-link interference from the DL AP to the UL AP as AP-AP interference (AAI). Our aim is to detect the data transmitted by the UL UEs to the UL AP, in the presence of the AAI caused by the DL AP.

We assume a block-fading model where all channels stay constant throughout each coherence interval. Suppose that the APs are static in space and the UEs are mobile, such that the AP-AP channel varies much more slowly over time (i.e., has a longer coherence time) than the UE-AP channels. Assuming that the coherence bandwidths of all channels are of the same order, the AP-AP channel will have a larger coherence interval length than the UE-AP channels. Concretely, let $\tau$ (samples) be the coherence interval length of the UE-AP channels and assume that the AP-AP channel stays constant throughout $T$ realizations $t \in \TT \triangleq \{ 1,\dots,T \}$ of the UE-AP channels.

We define the collective channel matrix from the $K$ UL UEs to the UL AP in coherence interval $t \in \TT$ as
\begin{align}
    \G_t = 
    \begin{bmatrix}
        \g_{1,t} & \dots & \g_{K,t}
    \end{bmatrix}
    \in \C^{M\times K},
\end{align}
where $\g_{k,t} \in \C^{M \times 1}$ is the channel vector from UL UE $k$ to the UL AP. Also, let 
\begin{align}
    \HH =
    \begin{bmatrix}
        \h_{1} & \dots & \h_{N}
    \end{bmatrix}
    \in \C^{M\times N}
\end{align}
be the interference channel from the DL AP to the UL AP, which is constant during all $t \in \TT$. Here, $\h_{n} \in \C^{M \times 1}$ is the channel from the $n$th DL AP antenna to the UL AP. We illustrate the coherence interval structure in Fig. \ref{Fig:block}.

Further, let $\X_t \in \C^{K \times \tau}$ be the collective matrix of data symbols transmitted from the $K$ UL UEs in coherence interval $t$. 
Similarly, let $\SSSS_t \in\C^{J \times \tau}$ be the data symbols from the DL AP intended for the $J$ DL UEs, and let $\V_t \in \C^{N \times J}$ be the corresponding precoding matrix. The \emph{effective} signal transmitted by the DL AP is $\V_t \SSSS_t \in \C^{N \times \tau}$.

The signal received by the UL AP in each coherence interval $t \in \TT$ can be expressed as
\begin{align}
    \label{Yt}
    \Y_t \triangleq \G_t \X_t + \HH \V_t \SSSS_t + \W_t \in \C^{M \times \tau},
\end{align}
where $\W_t$ is the receiver noise. The transmit power and power control coefficients are embedded in the UL channel matrices $\{ \G_t \}$ and the DL data matrices $\{ \SSSS_t \}$. We make an assumption.
\begin{assumption}
    \label{assum:UEAPchannels}
    The UL AP knows the UE-AP channels $\{ \G_t \}$, the DL data matrices $\{ \SSSS_t \}$ and the DL precoding matrices $\{ \V_t \}$, but not the AP-AP channel $\HH$.
\end{assumption}
In practice the UE-AP channels $\{ \G_t \}$ must be estimated in each coherence interval, which can be accomplished by inserting known pilot sequences into some columns of $\{ \X_t \}$. We assume that $ \{ \G_t \} $ are known to simplify the analysis and to let us focus solely on the AAI mitigation problem. Further, we assume that the two APs are connected through a backhaul network such that $\{ \SSSS_t \}$ and $\{ \V_t \}$ can be shared.

\section{Problem Formulation}
\label{sec:problem}

Here, we define the UL data detection\footnote{Herein, by ``data detection'' we mean obtaining soft estimates of the data symbols. These estimates
can then, in a standard manner, be used to obtain per-bit likelihood ratios
for decoding.} and AAI mitigation problem. We leverage the slow time variations of the AP-AP channel $\HH$ to enhance the detection of the UL data $\{ \X_t \}$. One possible approach would be to transmit designated AP-AP pilot sequences to estimate $\HH$ over multiple UE-AP coherence intervals, and use this estimate to mitigate the interference. However, such methods require a training overhead which harms the UL and DL data rates as the data transmission must be halted during the channel estimation phase. To overcome this, we use the known DL transmit signals $\{ \V_t \SSSS_t \}$ to enable AP-AP channel estimation without any training overhead. Specifically, we propose to detect the UL data and mitigate the AAI by solving the least-squares problem
\begin{align}
    \label{LS}
    \min_{ \{\X_t\}, \HH } \ \sum_{t \in \TT} \| \Y_t - \G_t \X_t - \HH \V_t \SSSS_t \|^2,
\end{align}
which is a joint estimation of the UL data matrices $\{\X_t\}$ and the AP-AP channel $\HH$ over all $T$ UE-AP coherence intervals $t \in \TT$, during which $\HH$ is constant. To solve \eqref{LS}, we first rewrite it in vector form. From \cite[Lemma 4.3.1]{horn1991topics}, for matrices $\GAMMA$, $\DELTA$, and $\LAMBDA$ of compatible dimensions, we have 
\begin{align}
    \label{veckron}
    \vecc(\GAMMA\DELTA\LAMBDA) = (\LAMBDA^T\otimes \GAMMA)\vecc(\DELTA).
\end{align}
Let $\y_t \triangleq \vecc(\Y_t)$, $\x_t \triangleq \vecc(\X_t)$, $\h \triangleq \vecc(\HH)$, and $\w_t \triangleq \vecc(\W_t)$. Then, we can apply \eqref{veckron} to obtain the vector form of the received signal $\Y_t$ in \eqref{Yt} as
\begin{align}
    \label{yt}
    \y_t \triangleq 
    \begin{bmatrix}
        (\II_{\tau}\otimes \G_t) & (\SSSS_t^T \V_t^T \otimes \II_{M})
    \end{bmatrix} 
    \begin{bmatrix}
        \x_t \\ \h
    \end{bmatrix}
    + \w_t.
\end{align}
To formulate the joint estimation problem over all $t \in \TT$, we gather all received signals $\{ \y_t\}$ in the vector
\begin{align}
    \label{y}
    \y \triangleq
    \begin{bmatrix}
        \y_1 \\ \vdots \\ \y_T
    \end{bmatrix} 
    = \A\z + \w \in \C^{TM\tau \times 1},
\end{align}
where the regression matrix $\A \in \C^{TM\tau \times (TK\tau + MN)}$ is
\begin{align}
    \label{A}
    \A & \triangleq
    \begin{bmatrix}
        \II_{\tau}\otimes \G_1 & & \0 & \SSSS_1^T \V_1^T \otimes \II_{M}
        \\
        & \ddots & & \vdots
        \\
        \0 & & \II_{\tau}\otimes \G_T & \SSSS_T^T \V_T^T \otimes \II_{M}
    \end{bmatrix}
\end{align}
and
\begin{align}
    \z & \triangleq
    \begin{bmatrix}
        \x
        \\
        \h
    \end{bmatrix}, \quad
    \x \triangleq
    \begin{bmatrix}
        \x_1 \\ \vdots \\ \x_T
    \end{bmatrix}, \quad
    \w \triangleq
    \begin{bmatrix}
        \w_1 \\ \vdots \\ \w_T
    \end{bmatrix}.
\end{align}
Now, the original least-squares problem \eqref{LS} can be expressed on the equivalent vector form
\begin{align}
    \label{LSvec}
    \min_{\z} \ \| \y - \A \z \|^2.
\end{align}

\begin{figure}[t!]
	\centering
	\begin{tikzpicture}[scale=0.7]
		\draw[thick] (2,0) -- (10,0);
		\draw[thick] (2,-1) -- (10,-1);
		\draw[thick] (2,0) -- (2,-1);
		\draw[thick] (10,0) -- (10,-1);
		\draw[thick] (4,0) -- (4,-1);
		\draw[thick] (6,0) -- (6,-1);
		\draw[thick] (8,0) -- (8,-1);
		\draw[thick] (10,0) -- (10,-1);

		\draw [decorate,
		decoration = {calligraphic brace, raise=3pt, amplitude=3pt}, thick] (2.05,0) --  (3.95,0) node[pos=0.5,above=5pt,black,align=center, text width=100pt]{One UE-AP coherence interval ($\tau$ samples)};
		
		\node at (-0.1,-0.5) {UE-AP channels:};
		\node at (3,-0.5) {$\G_1$};
		\node at (5,-0.5) {$\G_2$};
		\node at (7,-0.5) {$\dots$};
		\node at (9,-0.5) {$\G_{T}$};
		
		\draw[thick] (2,-2) -- (10,-2);
		\draw[thick] (2,-1) -- (2,-2);
		\draw[thick] (10,-1) -- (10,-2);
		\draw[thick] (4,-1) -- (4,-2);
		\draw[thick] (6,-1) -- (6,-2);
		\draw[thick] (8,-1) -- (8,-2);
		\draw[thick] (10,-1) -- (10,-2);
		
		\node at (-0.1,-1.5) {AP-AP channel:};
		\node at (3,-1.5) {$\HH$};
		\node at (5,-1.5) {$\HH$};
		\node at (7,-1.5) {$\dots$};
		\node at (9,-1.5) {$\HH$};
		
            \draw [decorate, decoration = {calligraphic brace, raise=3pt, amplitude=3pt, mirror}, thick] (2.05,-2) -- (9.95,-2) node[pos=0.5,below=5pt,black,align=center, text width=120pt]{One AP-AP coherence interval ($T\tau$ samples) };

            \draw[thick, ->] (2,-4) -- (10,-4) node[pos=0.5, below=0pt]{Time};
		
	\end{tikzpicture}
	\caption{The UE-AP and AP-AP channels during a sequence of $T$ UE-AP coherence intervals, during which the AP-AP channel is constant.}
	\label{Fig:block}
\end{figure}
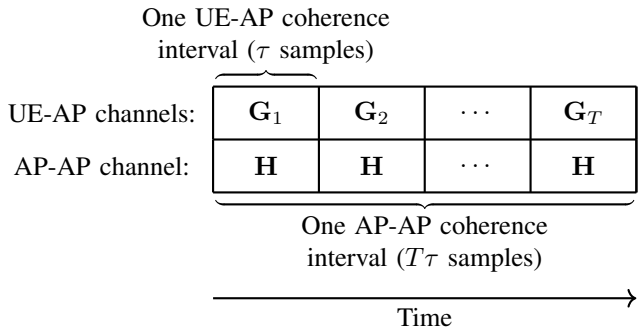

If $\A$ has full column rank, \eqref{LSvec}  has the unique solution
\begin{align}
    \label{lszhat}
    \hat\z \triangleq
    \begin{bmatrix}
        \hat\x
        \\
        \hat\h
    \end{bmatrix}
    = \argmin_{\z} \| \y - \A\z \|^2 = \left( \A^H \A \right)^{-1} \A^H \y.
\end{align}
However, in general $\A$ may be rank-deficient.\footnote{The intuition is that the unknown matrix $\HH$ is observed over the rank-deficient matrix $\V_t \SSSS_t$ for each $t$; see \eqref{Yt}. That is, to estimate all dimensions of $\HH$, we must observe it over sufficiently many independent realizations of $\{ \V_t \SSSS_t \}$. The coupling with $\{ \X_t \}$ makes this even more complicated.} In case $\A$ is rank-deficient, it is possible to reduce the dimension of the null space by inserting zeros into $\x$ (i.e., by forcing the UL UEs to be silent in certain samples), which effectively removes columns from $\A$. The conditions under which it is possible to remove sufficiently many columns to make $\A$ full-rank are non-obvious
in general (for arbitrary $T$). In Section \ref{sec:datadet2}, we will use this 
zero-insertion technique to obtain a uniquely solvable least-squares problem for rank-deficient $\A$ in the special case of $T=1$.

The rest of the paper is structured as follows. Section \ref{sec:analysis} is concerned with the existence of a unique solution to \eqref{LSvec}. In Section \ref{sec:datadet} we analyze the statistics of the unique solution \eqref{lszhat}, and develop remedies for the rank-deficient case. Finally, in Section \ref{sec:numerical}, we evaluate our proposed methods numerically.

\section{Uniqueness of the solution to \eqref{LSvec}}
\label{sec:analysis}

Finding conditions for uniqueness of the solution to \eqref{LSvec} is vital. If the least-squares problem lacks a unique solution, the estimation problem is unidentifiable and infinitely many sets of data matrices $\{\X_t\}$ and AP-AP channels $\HH$ explain the same received signal \eqref{Yt}, making it impossible to recover the data matrices even in the noise-free case. Therefore, uniqueness conditions determine when our joint data detection and AAI mitigation technique is meaningful to apply.

The solution \eqref{lszhat} is unique if $\A$ is a tall or square matrix with full rank. We assume that $\A$ is tall or square.
\begin{assumption}
    \label{assum:A}
    The coherence interval length $\tau$ is sufficiently large, such that
    the matrix $\A$ in \eqref{A} is tall or square, i.e.,
    \begin{align}
        \label{eq:tau}
        TM\tau \geq TK\tau + MN \iff \tau \geq \dfrac{MN}{T(M-K)}.
    \end{align}
    This is a reasonable assumption which will be fulfilled in practice; see Section \ref{sec:taudisc}.
\end{assumption}
However, $\A$ does not always have full rank. 
Next, we will analyze the rank of $\A$. First, we derive a necessary condition for full-rank $\A$ in the general case. Thereafter, we consider two special cases where we derive sufficient full-rank conditions or calculate the exact rank of $\A$.

\vspace{-3mm}
\subsection{Analyzing the Rank of $\A$}

Our further analysis will rely on one critical assumption.
\begin{assumption}
    \label{assum:VS}
    All elements of $\{ \G_t \}$, $\{ \V_t \}$ and $\{ \SSSS_t \}$ are mutually independent random variables with a positive probability density on the entire complex plane.
\end{assumption}
From Assumption \ref{assum:VS}, it follows that all $\{ \G_t \}$, $\{ \V_t \}$ and $\{ \SSSS_t \}$ have full rank with probability one, and that for the effective DL transmit signals $\{ \V_t \SSSS_t \}$ it holds: 
\begin{align}
    \label{rankPt}
    \rank(\V_t \SSSS_t) & = \min \{ \rank(\V_t),\ \rank(\SSSS_t) \} = J.
\end{align}

Next, we calculate a lower bound on the rank of $\A$, and derive a necessary condition on $T$ for full-rank $\A$.
\begin{theorem}
    \label{theorem:dimNA}
    The rank of $\A$ is (with probability one)\footnote{Hereafter, all statements on the
    rank of matrices are with probability one.}:
    \begin{align}
        \label{eq:rankA}
        \rank(\A) = TK\tau + R,
    \end{align}
    for a non-negative constant $R \leq MN$. Thus, the dimension of the null space $\NN(\A)$ of $\A$ is: 
    \begin{align}
        \label{eq:dimNA}
        \dim(\NN(\A)) & = TK\tau + MN -\rank(\A) = MN - R.
    \end{align}
    Furthermore, for the matrix $\A$ to have full rank, the following necessary condition must hold:
    \begin{align}
        \label{eq:boundT}
        T \geq \dfrac{MN}{(M-K)J}.
    \end{align}
\end{theorem}
\begin{proof}
    See Appendix \ref{appendix:dimNA}.
\end{proof}
Note that \eqref{eq:boundT} is a \emph{necessary but not sufficient} condition for full-rank $\A$. 
Below, we will consider a special case where a sufficient condition can also be established. We have to leave the problem of calculating the exact rank of $\A$ in the general case as an open question. 

Henceforth, we will limit our analysis to two cases:
\begin{itemize}
    \item \textbf{Case 1 ($M=N$, $K=J$):} Arbitrary $T$, but we restrict the number of UL/DL AP antennas and the number of UL/DL UEs, respectively, to be equal. We derive a lower bound on $T$ that ensures full-rank $\A$ (Theorem \ref{theorem:M=N=2K=2J} below).
    \item \textbf{Case 2 ($T=1$):} Arbitrary $M,N,K,J$, but the UE-AP channels vary as slowly as the AP-AP channel ($T=1$). We calculate the exact rank of $\A$ (Theorem \ref{theorem:T=1} below).
\end{itemize}

\subsection{Case 1 ($M=N$, $K=J$, arbitrary $T$)}

Calculating the exact rank of $\A$ is a challenging problem even in the case of $M=N$, $K=J$. We will prove that the necessary bound \eqref{eq:boundT} is also sufficient in a special case.
\begin{theorem}
    \label{theorem:M=N=2K=2J}
    When $M=N=2K=2J$, i.e. with twice as many AP antennas as UEs, $\A$ has full rank if and only if
    \begin{align}
        \label{eq:T}
        T \geq 4.
    \end{align}
\end{theorem}
\begin{proof}
    See Appendix \ref{appendix:SC_MN_2KJ}.
\end{proof}
Note that the condition \eqref{eq:T} is consistent with \eqref{eq:boundT} by inserting $M = N = 2K = 2J$. We believe that the technique used in the proof of Theorem \ref{theorem:M=N=2K=2J} can be extended to prove that the necessary condition \eqref{eq:boundT} is also sufficient for general $M=N$, $K=J$. Extensive numerical simulations suggest that it is, but this remains an unproven conjecture.
\begin{conjecture}
    \label{conj:M=N,K=J}
    With $M=N$ and $K=J$, the matrix $\A$ has full rank if and only if
    \begin{align}
        \label{eq:conjboundT}
        T \geq \dfrac{M^2}{(M-K)K}.
    \end{align}
\end{conjecture}
The interpretation of Conjecture \ref{conj:M=N,K=J} is that $\A$ will eventually become full-rank for sufficiently large $T$.

\subsection{Case 2 ($T=1$, arbitrary $M,N,K,J$)}

With $T=1$, $\A$ is given by
\begin{align}
    \label{A_T1}
    \A = 
    \begin{bmatrix}
        \II_{\tau} \otimes \G_1 & \SSSS_1^T \V_1^T \otimes \II_M
    \end{bmatrix}
     \in \C^{M\tau \times (K\tau + MN)}.
\end{align}
We derive a closed-form expression for the rank of $\A$.
\begin{theorem}
    \label{theorem:T=1}
    With $T=1$, the rank of $\A$ in \eqref{A_T1} is:
    \begin{align}
        \rank(\A) = K\tau + (M-K)J.
    \end{align}
    Hence, since $\A$ is tall,
    \begin{align}
        \label{dimNA_case1}
        \nonumber
        \dim(\NN(\A)) & = K\tau + MN - \rank(\A)
        \\
        & = M(N - J) + KJ.
    \end{align}
\end{theorem}
\begin{proof}
    See Appendix \ref{appendix:SC_T1}.
\end{proof}
Theorem \ref{theorem:T=1} states that $\A$ \emph{never} has full rank when $T=1$; hence the solution to \eqref{LSvec} is not unique. In Section \ref{sec:datadet2}, we modify \eqref{LSvec} to make the problem uniquely solvable for $T=1$.

\section{Joint Data Detection and AAI Mitigation}
\label{sec:datadet}

In this section, we present methods for joint UL data detection and AAI mitigation. First, we calculate the covariance matrix of $\hat\x$ in \eqref{lszhat} assuming a full-rank matrix $\A$. Then, we propose a method to obtain a unique solution for rank-deficient $\A$. We show that our proposed methods are independent of the AAI power. Finally, we present two baseline methods.

Throughout, we assume that the elements of all $\W_t$ are independent and identically distributed (i.i.d.) $\CN(0,1)$, and that the AP-AP channel $\HH$ has full rank.

\vspace{-3mm}
\subsection{Proposed Method for Full-Rank $\A$}

When $\A$ has full rank, the covariance matrices of the (unique) data symbol estimates $\{\hat{\x}_t\}$ are as follows.
\begin{proposition}
    \label{prop:covxt}
    The covariance matrix $\Cov(\hat\x_t)$ of the UL data symbol estimate $\hat\x_t$ in coherence interval $t$ is given by the $t$th $K\tau \times K\tau$ diagonal block of the full covariance matrix $( \A^H \A )^{-1}$ of $\hat\z$. $\Cov(\hat\x_t)$ is provided in \eqref{Covxt}.
    \begin{figure*}[t!]
	\begin{align}
		\label{Covxt}
	    \Cov(\hat\x_t) & = \II_{\tau}\otimes (\G_t^H \G_t)^{-1} + (\SSSS_t^T \V_t^T \otimes (\G_t^H \G_t)^{-1} \G_t^H) \E (\V_t^* \SSSS_t^* \otimes \G_t (\G_t^H \G_t)^{-1}),
		\\
		\label{E}
		\E & \triangleq \Cov(\hat\h) = \left(\sum_{t \in \TT} \V_t^* \SSSS_t^* \SSSS_t^T \V_t^T \otimes (\II_M - \G_t (\G_t^H \G_t)^{-1} \G_t^H) \right)^{-1}.
	\end{align}
	\hrulefill
    \end{figure*}
\end{proposition}
\begin{proof}
	See Appendix \ref{appendix:covxt}.
\end{proof}
Note that the covariance matrix \eqref{Covxt} contains two terms. The first term arises due to the receiver noise at the UL AP. The second term comes from the imperfect AAI mitigation and depends on the covariance matrix of the AP-AP channel estimate in \eqref{E}. This is intuitive; a good AP-AP channel estimate has a small covariance matrix, which reduces the uncertainty in the data estimates. 

\vspace{-3mm}
\subsection{Proposed Method for Rank-Deficient $\A$ and $T=1$}
\label{sec:datadet2}

As alluded to in Section \ref{sec:problem}, we can reduce the dimension of the nullspace of $\A$ by silencing a subset of the UL data symbols. Here, we use this idea to obtain a uniquely solvable problem for $T=1$.\footnote{This method can be applied for any $T > 1$ by disregarding that $\HH$ is slowly time-varying. Better methods can possibly be developed for rank-deficient $\A$ with $T > 1$, but we leave this as an open problem.} 

For simplicity, we drop the index $t$ in the following. Before silencing UL data samples we can do another trick to reduce the dimension of the nullspace of $\A$ in \eqref{A_T1}. Specifically, we define the \emph{effective} AP-AP channel $\F \triangleq \HH \V \in \C^{M \times J}$ and treat this matrix as unknown instead of $\HH$. This effectively reduces the number of unknown channel parameters from $MN$ to $MJ$. To see how this affects the nullspace of $\A$, we write the received signal as
\begin{align}
    \label{YF}
    \Y = \G\X + \HH\V\SSSS + \W = \G\X + \F \SSSS + \W,
\end{align}
and the corresponding least-squares problem becomes
\begin{align}
    \label{LS_T1}
    \min_{ \X, \F } \ \| \Y - \G \X - \F \SSSS \|^2.
\end{align}
Define $\f \triangleq \vecc(\F)$ and rewrite \eqref{YF} on the form
\begin{align}
    \y & = \A' \z' + \w,
\end{align}
where $\A' \in \C^{M\tau \times (K\tau + MJ)}$ and $\z' \in \C^{(K\tau + MJ) \times 1}$ are
\begin{align}
    \label{Az}
    \A' =
    \begin{bmatrix}
        \II_{\tau} \otimes \G & \SSSS^T \otimes \II_M
    \end{bmatrix}
    ,\
    \z' =
    \begin{bmatrix}
        \x
        \\
        \f
    \end{bmatrix}.
\end{align}
By following the same idea as in the proof of Theorem \ref{theorem:T=1} it follows immediately that $\dim(\NN(\A')) = KJ$. Thus, by incorporating the precoding matrix $\V$ into the AP-AP channel and treating $\F$ as unknown, we have reduced the dimension of the nullspace of the regression matrix from $M(N-J) + KJ$ to $KJ$.

We will show that \eqref{LS_T1} becomes uniquely solvable if the $K$ UL UEs are silent in $J$ samples of each coherence interval. Without loss of generality, choose the first $J$ samples. Then,
\begin{align}
    \X =
    \begin{bmatrix}
        \0_{K \times J} & \mybar\X
    \end{bmatrix},
\end{align}
where $\mybar\X \in \C^{K \times (\tau - J)}$ contains the data transmitted in the remaining $\tau - J$ samples. Then, the UL AP receives the signal
\begin{align}
    \label{Y1_T=1}
    \Y = \G
    \begin{bmatrix}
        \0_{K \times J} & \mybar\X
    \end{bmatrix}
    + \F \SSSS  + \W.
\end{align}
Define $\mybar\x \triangleq \vecc(\mybar\X)$ and apply \eqref{veckron} to obtain the equivalent vector form of \eqref{Y1_T=1}:
\begin{align}
    \label{eq:y_T=1}
    \y & =
    \begin{bmatrix}
        \II_{\tau}\otimes\G & \SSSS^T  \otimes \II_{M}
    \end{bmatrix} 
    \begin{bmatrix}
        \0_{KJ \times 1} \\ \mybar\x \\ \f
    \end{bmatrix}
    + \w
    = \mybar\A \mybar\z + \w,
\end{align}
where $\mybar\A \in \C^{M\tau \times (K(\tau-J) + MJ)}$ and $\mybar\z \in \C^{(K(\tau-J) + MJ) \times 1}$ are given by
\begin{align}
    \mybar\A = 
    \begin{bmatrix}
        \0_{MJ \times K(\tau - J)} & |
        \\
        & \SSSS^T \otimes \II_{M}
        \\
        \II_{\tau-J} \otimes \G & |
    \end{bmatrix}
    ,\
    \mybar\z =
    \begin{bmatrix}
        \mybar\x \\ \f
    \end{bmatrix}.
\end{align}
Note that $\mybar\A$ is obtained by removing the first $KJ$ columns from $\A'$ in \eqref{Az}. Now, a unique least-squares solution exists. The regression matrix $\mybar\A$ of the modified problem \eqref{eq:y_T=1} has full column rank.
\begin{theorem}
    \label{theorem:dimNAbar}
    The matrix $\mybar\A$ has full rank:
    \begin{align}
        \label{eq:rankAbar}
        \rank(\mybar\A) = K(\tau - J) + MJ.
    \end{align}
\end{theorem}
\begin{proof}
    See Appendix \ref{appendix:dimNAbar}.
\end{proof}
We present two algorithms to estimate $ \mybar\x$ from $\y$ in \eqref{eq:y_T=1}. 

\subsubsection{Algorithm 1 (Joint Estimation)}
\label{defA_T=1}

Here, in the same spirit as the joint estimation in \eqref{lszhat},
we directly obtain $\hat{\mybar{\x}}$ by solving the joint least-squares problem
\begin{align}
    \label{zBarHat}
    \hat{\mybar\z} \triangleq
    \begin{bmatrix}
        \hat{\mybar{\x}} \\ \hat\f
    \end{bmatrix}
    = \argmin_{\mybar\z} \| \y - \mybar\A \mybar\z \|^2 = \left( \mybar\A^H \mybar\A \right)^{-1} \mybar\A^H \y.
\end{align}
We calculate the covariance matrix of the data estimate $\hat{\mybar{\x}}$.
\begin{proposition}
    The covariance matrix $\Cov(\hat{\mybar\x})$ is given by the upper-left $K(\tau-J) \times K(\tau-J)$ block of the full covariance matrix $( \mybar\A^H \mybar\A )^{-1}$ of $\hat{\mybar\z}$. $\Cov(\hat{\mybar\x})$ is provided in \eqref{Covxtbar_T=1}, where $\mybar\SSSS$ is obtained by removing the first $J$ columns from $\SSSS$.
    \begin{figure*}[t!]
	\begin{align}
		\label{Covxtbar_T=1}
	    \Cov(\hat{\mybar\x}) & = \II_{\tau-J}\otimes (\G^H \G)^{-1} + (\mybar\SSSS^T \otimes (\G^H \G)^{-1} \G^H) \mybar\E (\mybar\SSSS^* \otimes \G (\G^H \G)^{-1}),
		\\
		\label{Ebar_T=1}
		\mybar\E & \triangleq \Cov(\hat\f) = \left( (\SSSS^* \SSSS^T \otimes \II_M) - (\mybar\SSSS^* \mybar\SSSS^T \otimes \G (\G^H \G)^{-1} \G^H) \right)^{-1}.
	\end{align}
	\hrulefill
    \end{figure*}
\end{proposition}
\begin{proof}
    The derivation follows the same steps as the proof of Proposition \ref{prop:covxt} and is thus omitted.
\end{proof}

\subsubsection{Algorithm 2 (Channel Estimation Followed by Data Detection)}
In Algorithm 1, the first $J$ samples are not explicitly used. Now, we use these samples to estimate the effective AP-AP channel $\F$, before subtracting the interference from the remaining samples and detecting the data. The signal received in the first $J$ samples, i.e., the first $J$ columns of $\Y$ in \eqref{Y1_T=1} is
\begin{align}
    \label{Y_1:J}
    \nonumber
    \Y_{1:J} & \triangleq \G
    \begin{bmatrix}
        \0_{K \times J} & \mybar\X
    \end{bmatrix}_{1:J}
    + \F \SSSS_{1:J} + \W_{1:J}
    \\
    & = \F \SSSS_{1:J} + \W_{1:J},
\end{align}
where $(\cdot)_{1:J}$ denotes the first $J$ columns of a matrix. The least-squares estimate of $\F$ from \eqref{Y_1:J} is:
\begin{align}
    \label{eq:hatF}
    \hat\F = \Y_{1:J} \SSSS_{1:J}^H (\SSSS_{1:J} \SSSS_{1:J}^H)^{-1}.
\end{align}
Now partition $\Y$ according to ${\Y = \begin{bmatrix} \Y_{1:J} & \mybar\Y \end{bmatrix}}$. Consider $\mybar\Y$, and subtract the known interference $\hat\F \mybar\SSSS$:
\begin{align}
    \label{Ybreve}
    \mybar\Y - \hat\F \mybar\SSSS & = \G \mybar\X + (\F - \hat\F) \mybar\SSSS + \mybar\W.
\end{align}
The least-squares estimate of $\mybar\X$ from \eqref{Ybreve} is
\begin{align}
    \label{XBarHat_C2}
    \hat{{\mybar\X}} = (\G^H \G)^{-1} \G^H ( \mybar\Y - \hat\F \mybar\SSSS ),
\end{align}
and on vector form,  $\hat{\mybar\x} = \vecc(\hat{{\mybar\X}})$. 

Despite the seemingly different nature of the two algorithms (joint estimation respectively channel estimation followed by data detection), we show that they have the same solution.
\begin{theorem}
    \label{theorem:eqMethods}
    The solutions $\hat{\mybar{\x}}$ of the two algorithms, obtained with joint estimation in \eqref{zBarHat} and with channel estimation followed by interference subtraction in \eqref{XBarHat_C2}, are equal.
\end{theorem}
\begin{proof}
    See Appendix \ref{appendix:eqMethods}.
\end{proof}

\vspace{-3mm}

\subsection{The Proposed Methods are Independent of the AAI Power}
\label{sec:indep}
Notably, the proposed methods (both for the full-rank case and the rank-deficient case) are independent of the DL SNR, i.e., they are independent of the interference power.

We can observe this by looking at the covariance matrices of the data estimates in \eqref{Covxt} and \eqref{Covxtbar_T=1}. The first term in the covariance matrices depends only on $\{ \G_t \}$ and is clearly independent of the DL SNR. Next, we focus on the second term of \eqref{Covxt}, where the DL SNR enters through $\{ \SSSS_t \}$, appearing twice in \eqref{Covxt} which causes a proportional increase. However, this is compensated for by the covariance matrix of the AP-AP channel estimate $\E$ in \eqref{E}, which is inversely proportional to the DL SNR. To make this mathematically precise, we express the data matrices $\{ \SSSS_t \}$ as $\SSSS_t = \sqrt{\rho_\text{d}} \tilde{\SSSS}_t$, where $\rho_\text{d}$ is the DL SNR and the elements of $\tilde{\SSSS}_t$ have unit variance. Thus, we have $\E = \tilde{\E} / \rho_{\text{d}}$, where 
\begin{align}
    \nonumber
    \tilde\E = \left(\sum_{t \in \TT} \V_t^* \tilde{\SSSS}_t^* \tilde{\SSSS}_t^T \V_t^T \otimes (\II_M - \G_t (\G_t^H \G_t)^{-1} \G_t^H) \right)^{-1}.
\end{align}
Now consider the second term of \eqref{Covxt}:
\begin{align}
    \label{MSEtrace}
    \nonumber
    & (\SSSS_t^T \V_t^T \otimes (\G_t^H \G_t)^{-1} \G_t^H) \E (\V_t^* \SSSS_t^* \otimes \G_t (\G_t^H \G_t)^{-1})
    \\
    \nonumber
    = & (\sqrt{\rho_\text{d}} \tilde{\SSSS}_t^T \V_t^T \otimes (\G_t^H \G_t)^{-1} \G_t^H) \tilde{\E} / \rho_{\text{d}}
    \\
    \nonumber
    & \qquad  \times (\V_t^* \sqrt{\rho_\text{d}} \tilde{\SSSS}_t^* \otimes \G_t (\G_t^H \G_t)^{-1})
    \\
    = & ( \tilde{\SSSS}_t^T \V_t^T \otimes (\G_t^H \G_t)^{-1} \G_t^H) \tilde{\E} (\V_t^*  \tilde{\SSSS}_t^* \otimes \G_t (\G_t^H \G_t)^{-1}),
\end{align}
where the final expression is independent of $\rho_{\text{d}}$. It follows that the second term of \eqref{Covxt}, and hence the full covariance matrix, is independent of the DL SNR. (The motivation for \eqref{Covxtbar_T=1} is analogous; we omit the details for brevity.) The intuition is that the channel estimate $\hat\h$ (or $\hat{\f}$) improves when increasing the DL SNR, which compensates for the UL data being observed in stronger AAI.

\subsection{Baseline Methods}
\label{sec:baselines}

We present two baselines to compare  our proposed methods against: a genie-aided detector with perfect AAI mitigation, and a worst-case detector without AAI mitigation.

\subsubsection{Genie-Aided Baseline (Perfect AAI Mitigation)}

We assume that the AP-AP channel $\HH$ is known to the UL AP. Thus, the AAI can be subtracted from the received signal \eqref{Yt} to get
\begin{align}
    \label{dotY}
    \dot\Y_t \triangleq \Y_t - \HH \V_t \SSSS_t = \G_t \X_t + \W_t \in \C^{M\times \tau}.
\end{align}
Let $\dot\y_t = \vecc(\dot\Y_t) \ \forall t \in \TT$, and collect these in
\begin{align}
    \label{ydot}
    \dot\y \triangleq
    \begin{bmatrix}
        \dot\y_1 \\ \vdots \\ \dot\y_T
    \end{bmatrix}
    = \dot\A \x + \w \in \C^{TM\tau \times 1},
\end{align}
where $\dot\A \in \C^{TM\tau \times TK\tau}$ is given by
\begin{align}
    \label{Adot}
    \dot\A & =
    \begin{bmatrix}
        \II_{\tau}\otimes \G_1 & & \0
        \\
        & \ddots &
        \\
        \0 & & \II_{\tau}\otimes \G_T
    \end{bmatrix}.
\end{align}
The matrix $\dot\A$ clearly has full column rank. Thus, the unique least-squares estimate of $\x$ from \eqref{ydot} is
\begin{align}
    \label{eq:xg}
    \nonumber
    \hat\x_{\text{g}} & \triangleq \argmin_{\x} \| \dot\y - \dot\A \x \|^2 = (\dot\A^H \dot\A)^{-1} \dot\A^H \dot\y
    \\
    & = (\dot\A^H \dot\A)^{-1} \dot\A^H (\dot\A \x + \w) = \x + \dot\w,
\end{align}
where
\begin{align}
    \dot\w \triangleq (\dot\A^H \dot\A)^{-1} \dot\A^H \w \sim \CN(\0_{TK\tau \times 1}, (\dot\A^H \dot\A)^{-1}).
\end{align}
The covariance of $\hat\x_{\text{g}}$ is $\Cov(\hat\x_{\text{g}}) = \Cov(\dot\w)=(\dot\A^H \dot\A)^{-1}$.

\subsubsection{Baseline Without AAI Mitigation}
Here, we proceed as if we received the interference-free signal \eqref{ydot}, although we actually received the signal with interference in \eqref{y}. The estimate with no AAI mitigation, $\hat\x_{\text{n}}$, is given by
\begin{align}
    \label{eq:xhatn}
    \nonumber
    \hat\x_{\text{n}} & \triangleq (\dot\A^H \dot\A)^{-1} \dot\A^H \y = (\dot\A^H \dot\A)^{-1} \dot\A^H (\A \z + \w)
    \\
    \nonumber
    & = (\dot\A^H \dot\A)^{-1} \dot\A^H
    \begin{bmatrix}
        | & \SSSS_1^T \V_1^T \otimes \II_{M}
        \\
        \dot\A & \vdots
        \\
        | & \SSSS_T^T \V_T^T \otimes \II_{M}
    \end{bmatrix}
    \begin{bmatrix}
        \x \\ \h
    \end{bmatrix}
    + \dot\w
    \\
    & =
    \x +
    \underbrace{
    \begin{bmatrix}
        \SSSS_1^T \V_1^T \otimes (\G_1^H\G_1)^{-1}\G_1^H
        \\
        \vdots
        \\
        \SSSS_T^T \V_T^T \otimes (\G_T^H\G_T)^{-1}\G_T^H
    \end{bmatrix}
    }_{\triangleq \B}
    \h
    + \dot\w.
\end{align}
This estimate has the same covariance as the genie-aided estimate, i.e. $\Cov(\hat\x_{\text{n}}) = \Cov(\hat\x_{\text{g}})$,
but it is biased:
\begin{align}
    \Bias(\hat\x_{\text{n}}) = \B \h,
\end{align}
where $\B$ was defined in \eqref{eq:xhatn} above.

\section{Practical Considerations}

In this section, we will discuss some practical aspects of our  models and assumptions. 

\vspace{-3mm}
\subsection{Feasible Values of $T$ in Practical Systems}
\label{sec:Tdisc}

In general, for a wireless channel, the coherence interval length is given by $T_{\text{c}} B_{\text{c}}$, where $T_{\text{c}}$ is the coherence time and $B_{\text{c}}$ is the coherence bandwidth. 
In what follows, we assume that $B_{\text{c}}$ is the same for all channels (both AP-AP and UE-AP channels); however, $T_{\text{c}}$ will differ.

Let $T_{\text{c,UE-AP}}$ and $T_{\text{c,AP-AP}}$ be the coherence times of the UE-AP channels and the AP-AP channel, respectively. Then, 
$\tau = T_{\text{c,UE-AP}} B_\text{c}$, $T \tau = T_{\text{c,AP-AP}} B_\text{c}$, and $T = T_{\text{c,AP-AP}} / T_{\text{c,UE-AP}}$.  Following the two-path model \cite[Eq. (2.3)]{ngo16} the coherence time is $T_\text{c} = \lambda / (2v)$, where $\lambda$ is the carrier wavelength and $v$ is the UE velocity. In an urban environment, the velocities of UEs will, roughly, vary from $3$ km/h (pedestrian) to $60$ km/h (vehicle). With a carrier frequency of $2$ GHz ($\lambda = 0.15$ m), we have $T_{\text{c,UE-AP}} \approx 90$ ms for a pedestrian and $T_{\text{c,UE-AP}} \approx 4.5$ ms for a vehicle. Although the two APs are static, the AP-AP channel still varies slowly over time due to dynamics in the surroundings. A conservative assumption is that the AP-AP channel varies at the same rate as if one of the APs was moving with a velocity of $3$ km/h, yielding a coherence time of $T_{\text{c,AP-AP}} \approx 90$ ms. Thus, in an urban environment, we can conservatively assume $T_{\text{c,AP-AP}} / T_{\text{c,UE-AP}} = T \in \{ 1, \dots, 20 \}$. In high-mobility cases, one can argue for larger values of $T$.

In Table \ref{table:T}, we investigate whether these values of $T$ are sufficient to make $\A$ full-rank, by presenting the $T$ necessary to fulfill the conjectured lower bound \eqref{eq:conjboundT}. For $T \geq 10$, $\A$ has full rank in all examples. However, a specific value of $T$ cannot be ensured in practice; thus the rank-deficient case must also be considered. 

\begin{table}[t!]
    \caption{The necessary values of $T$ and $\tau$ to make the matrix $\A$ tall and full-rank for some network setups}
    \label{table:T}
    \centering
    \begin{tabular}{c c | c | c} \toprule
        {AP antennas} & {Number of UEs} & {Necessary} & {Necessary} \\
        $M=N$ & $K=J$ &  $T \geq$ & $\tau \geq$ \\ \midrule
        8 & 4 & 4 & 4 \\
        16  & 4 & 6 & 4 \\
        32 & 4  & 10 & 4 \\
        32 & 8 & 6 & 8 \\
        64 & 8 & 10 & 8 \\
        64 & 16 & 6 & 15 \\
        \bottomrule
    \end{tabular}
\end{table}

\vspace{-3mm}
\subsection{Latency and Memory Growth}
\label{sec:latency}

When jointly processing the UL data over $T$ coherence intervals, latency increases  since the data is only decoded at the end of each AP-AP coherence interval. The latency is determined by the AP-AP coherence time $T_\text{c,AP-AP}$. For the example in Section \ref{sec:Tdisc}, the worst-case latency (for the first block of data) is $90$ ms, which is acceptable in delay-tolerant services.  
Since $T$ is a design variable upper-bounded by the channel characteristic and lower-bounded by rank constraints, we can choose any value in this range. That is, although the channel characteristics allow a certain value of $T$, one can choose to use a lower value in the implementation to fulfill constraints of latency-sensitive services. In extremely latency-critical scenarios, we can always use our proposed method for the rank-deficient case where $T=1$ as a fallback approach.

Further, the memory demand for storing $\{ \Y_t \}$, $\{ \G_t \}$, and $\{ \V_t \SSSS_t \}$ at the UL AP (as required for our proposed full-rank method) grows linearly with $T$. Specifically, the number of complex-valued variables that must be saved in memory to jointly process $T$ coherence intervals is $T(M\tau + MK + N\tau)$.

\vspace{-3mm}
\subsection{Perfectly Known $\{ \G_t \}$, $\{ \SSSS_t \}$, and $\{ \V_t \}$ at the UL AP}

	In practice, the UE-AP channels $\{ \G_t \}$ must be estimated in each UE-AP coherence interval by transmitting orthogonal pilots from the UL UEs. This induces an overhead of (at least) $K$ samples per UE-AP coherence interval, i.e., the channel estimation occupies a fraction $K/\tau$ of the time-frequency resources. Importantly, these channels must be estimated in any massive MIMO system; that is, our AAI mitigation methods do not increase the UE-AP channel estimation overhead beyond what is required in a standard static TDD setup. With imperfect channel knowledge, our proposed methods and baselines can be applied by inserting the estimates $\{ \hat\G_t \}$ instead of the true channels $\{ \G_t \}$ into the regression matrices $\A$, $\mybar\A$, and $\dot\A$, yielding mismatched least-squares solutions. The performance degradation caused by the mismatched detection depends on the quality of the channel estimates.
	
	Further, the assumption of known $\{ \V_t \}$ and $\{ \SSSS_t \}$ at the UL AP requires the DL AP to transfer $NJ + J\tau$ complex-valued symbols in each UE-AP coherence interval.
    In practice it is sufficient to transfer the effective DL signals $\{ \V_t \SSSS_t \}$ to the UL AP, in which case the number of complex-valued symbols is $N\tau$. Thus, over $T$ coherence intervals, the DL AP must transfer a total of $T \cdot \min \{ J(N+\tau), N\tau \}$ complex-valued symbols to the UL AP. These symbols are transferred via a backhaul network, and our methods are reliant on having a backhaul network with sufficiently high capacity to transfer these symbols to the UL AP in a timely manner.
    
    For AP-AP transfer, the transport delay (to transfer $\{ \V_t \SSSS_t \}$ to the UL AP over the backhaul) is on the order of a few to tens of milliseconds in practical systems \cite[Table 6.1-1]{tr36932}. In a centralized architecture, the transport delay can be reduced further. In any case, the delay is relatively small compared to the latency introduced by the joint processing over $T$ UE-AP coherence intervals. Thus, the backhaul delay will typically not be a limiting factor for our methods. If necessary, one can reduce the number of UE-AP coherence intervals to jointly process, i.e., operate with a smaller value of $T$ as discussed in Section \ref{sec:latency}, to compensate for the transport latency.

\vspace{-3mm}
\subsection{How Large $\tau$ is Required to Fulfill Assumption \ref{assum:A}?}
\label{sec:taudisc}

In Table \ref{table:T}, we present the lower bound on $\tau$ required to fulfill Assumption \ref{assum:A}, for the smallest value of $T$ that is necessary for full-rank $\A$. We note that these values are very moderate; hence we can expect Assumption \ref{assum:A} to be fulfilled in all practical scenarios. In fact, with $M=N$, the lower bound \eqref{eq:tau} is $\tau \geq M^2 / T (M-K)$. By inserting \eqref{eq:conjboundT} we get
\begin{align}
    \dfrac{M^2}{T (M-K)} \leq \dfrac{(M-K)K}{M^2} \dfrac{M^2}{(M-K)} = K.
\end{align}
That is, with $M=N$, the lower bound on $\tau$ is upper bounded by $K$, i.e., $\tau \geq K$ is enough to fulfill Assumption \ref{assum:A}.

\subsection{Computational Complexity}

A na\"ive approach to calculate the least-squares solution in the full-rank case is to directly invert $\A^H \A$, requiring a computational cost of $\OO((TK\tau + MN)^3)$, which scales cubically in $T$. However, by exploiting the bordered block-diagonal structure of $\A$ in \eqref{A} via the Schur complement, the cost can be reduced to $\OO\big( T \cdot ((MN)^2 + (M+N)^2 \tau) + (MN)^3 \big)$, i.e.\ linear in $T$ plus a one-time $\OO((MN)^3)$ cost for solving a reduced $MN \times MN$ system in $\h$. A detailed algorithm and complexity analysis is provided in Appendix \ref{appendix:complexity}.

\vspace{-3mm}
\subsection{Numerical Stability and Robustness}

We solve the least-squares problem after a standard column rescaling to unit norm, which yields a well-conditioned system. Let $\D$ be a diagonal matrix whose diagonal elements are the norms of the corresponding columns of $\A$. Then, we can rewrite the nominal problem \eqref{LSvec} as
\begin{align}
    \argmin_{\z} \| \y - \A \z \|^2 = \argmin_{\z} \| \y - \A \D^{-1} \D \z \|^2,
\end{align}
which can be solved for $\D \z$ using the scaled regression matrix $\A \D^{-1}$ with unit norm columns. The solution $\hat\z$ is recovered by multiplying the solution for $\D \z$ by $\D^{-1}$.

Consider the setup used later in our numerical evaluations (see Section \ref{sec:setup} for details), with $M=N=32$, $K=J=16$, and $\tau = 50$. All channels are i.i.d. Rayleigh fading and the DL AP performs maximum-ratio precoding, such that Assumption~\ref{assum:VS} is fulfilled. For $T=10$, numerical evaluations show that the median condition number of the scaled regression matrix $\A \D^{-1}$ is $\approx 20$. 

Further evaluations show that our methods are robust to scenarios where the strict independence in Assumption \ref{assum:VS} is relaxed. Specifically, when the DL AP applies zero-forcing instead of maximum-ratio precoding (which breaks the independence of the elements in $\{ \V_t \}$), the matrix $\A$ remains full-rank, and the median condition number of the rescaled matrix $\A\D^{-1}$ increases only slightly, from $\approx 20$ to $\approx 30$. Breaking the independence of $\{ \G_t \}$ likewise keeps $\A$ full-rank while increasing its condition number. For correlated Rayleigh fading channels with a Gaussian angular distribution with standard deviation $15^\circ$ \cite[Eq. (2.24)]{emil17}, the rescaled condition number increases to $\approx 100$. In the extreme case with an angular standard deviation of $0^\circ$, the channel matrices $\{ \G_t \}$ are rank-one and $\A$ becomes rank-deficient.

These results indicate that the sufficiency condition in Theorem \ref{theorem:M=N=2K=2J} may extend beyond Assumption \ref{assum:VS}, except in degenerate cases where $\A$ loses rank (such as the $0^\circ$ scenario above). A comprehensive numerical study (and theoretical extension) is required to make this precise, which we leave for future work.

\vspace{-3mm}
\subsection{The Block-Fading Model in our Context}

The analysis in this paper, and specifically the signal model in \eqref{Yt}, relies on a strict block-fading assumption where all channels are \emph{exactly} static throughout each coherence interval. This requires the definition of ``coherence'' to be rather conservative,  using a value of $\tau$ small enough that the variation of the channel within an interval is negligible.

Consider the example in Section \ref{sec:Tdisc}, with $T_\text{c,AP-AP} \approx 90$ ms and $T_\text{c,UE-AP} \in [4.5,\dots,90]$ ms. Combined with a typical outdoor coherence bandwidth $B_\text{c} \approx 300$ kHz \cite[Table 2.1]{ngo16}, we get an UE-AP coherence interval length in the range from $\tau \approx 1350$ (vehicle UE) to  $\tau \approx 27000$ (pedestrian UE). By using a conservative value, say $\tau = 50$, which is one to three orders of magnitude smaller than the coherence interval lengths above, then the block-fading model is a good approximation and the variations within one interval are negligible. As seen in Section \ref{sec:taudisc}, our methods are applicable for small values of $\tau$. 

However, it should be noted that in many other cases (e.g., ergodic capacity bounding \cite{ngo16}) one can define larger coherence intervals since the channel responses can be interpolated \emph{between coherence intervals}, which is not possible for our proposed algorithms.\footnote{For example, \cite{ngo16} allows for a channel phase rotation of $\pi/2$ within a coherence interval, which corresponds to Nyquist sampling of the time-varying frequency response and therefore allows error-free interpolation, see \cite[Sec.~2.1.4]{ngo16}. } Developing versions of our algorithms that allow for continuous variation of the channel response is an interesting topic for future work.  One possible direction is to model the time-frequency variations within a coherence interval by introducing a low-dimensional basis expansion framework, following the lines of \cite{bai2024JSTSP}.

\section{Numerical Results}
\label{sec:numerical}
In this section, we evaluate our proposed methods numerically. First, we calculate the MSEs of the data symbol estimates achieved by our proposed methods, and compare them to those obtained by the two baselines (see Section \ref{sec:baselines}). Then, we derive achievable SEs for our methods and compare them to the SEs obtained by the baselines and a state-of-the-art zero-forcing based AAI mitigation technique.

\subsection{Simulation Setup}
\label{sec:setup}
We consider the setup from Fig. \ref{Fig:system} with $M=N=32$ co-located UL/DL AP antennas and $K=J=16$ single-antenna UL/DL UEs. The two APs are located $500$ m apart. All UL and DL UEs are distributed uniformly at random within a circle of radius $250$ m around the respective AP. We use $\tau = 50$ which ensures Assumption \ref{assum:A}. The UL SNR $\rho_{\text{u}}$ is chosen such that the median SNR observed at the receiver side when transmitting with full power from a user at the cell-edge is $10$ dB, and $\rho_{\text{d}}$ is chosen such that the corresponding SNR from the DL AP to the UL AP is $20$ dB. 

We assume i.i.d. Rayleigh fading channels and include the UL SNR in the UL channels, i.e., they are generated as
\begin{align}
    \g_{k,t} & \sim \CN (\0_{M \times 1}, \rho_{\text{u}} \eta_{k,t} \beta_{k,t} \II_M)\ \forall k,t
    \\
    \h_{n} & \sim \CN (\0_{M \times 1}, \alpha \II_M)\ \forall n,
\end{align}
where $\beta_{k,t}$ and $\alpha$, respectively, are the large-scale fading coefficients from UL UE $k$ and the DL AP to the UL AP. Also, we introduced the UL power control coefficients $\{ \eta_{k,t} \}$. Further, assume that the channels between the DL UEs and the DL AP are i.i.d. Rayleigh fading and known at the DL AP. We let the DL AP perform maximum-ratio precoding. Thus, the DL precoding matrices are
\begin{align}
    \V_t & = 
    \begin{bmatrix}
        \vv_{1,t} & \dots & \vv_{J,t}
    \end{bmatrix}
    \\
    \vv_{j,t} & \sim \CN(\0_{N \times 1}, \theta_{j,t} \gamma_{j,t} \II_N / \Xi),
\end{align}
where $\gamma_{j,t}$ is the large-scale fading coefficient from the DL AP to DL UE $j$ in coherence interval $t$ and $\{ \theta_{j,t} \}$ are the corresponding DL power control coefficients. The normalization coefficient $\Xi$ of the covariance matrix is chosen to ensure the DL power constraint $\EEE \{ \trace \{ \V_t^H \V_t \}  \} = 1$. All large-scale fading coefficients $\{ \beta_{k,t} \}$, $\alpha$, and $\{ \gamma_{j,t} \}$ are calculated from the outdoor 3GPP Urban Microcell model \cite[Eqs. (37), (38)]{emil20TWC}, with elevations of 1 m and 10 m for the UEs and APs, respectively. The UL and DL power control coefficients are chosen as the ones achieving max-min fairness in a single-cell (static TDD) setup.

We let the DL data symbols, i.e. the elements of $\{ \SSSS_t \}$, be i.i.d. $\CN(0, \rho_{\text{d}})$. With these models, Assumption \ref{assum:VS} is fulfilled.

\subsection{Evaluation of Mean-Square Errors}

Here, we evaluate our proposed methods and compare them to the baselines in terms of their MSEs per data symbol.

\subsubsection{Calculating MSEs}
\label{sec:MSE}
The MSE of an estimator $\hat\x$ of $\x$ is
\begin{align}
    \label{mse}
    \MSE(\hat\x) & = \EEE \{ \| \hat\x - \x \|^2 \} = \trace(\Cov(\hat\x)) + \|\Bias(\hat\x)\|^2.
\end{align}
The bias term is non-zero only in the baseline without AAI mitigation. Normalizing \eqref{mse} by the dimension of $\x$ yields the following MSE expressions, conditioned on the known effective DL signals $\{ \V_t \SSSS_t \}$ (we omit the derivations):
\begin{itemize}
    \item Genie-Aided Baseline:
    \begin{align}
        \MSE(\hat\x_{\text{g}}) = \dfrac{1}{TK} \sum_{t \in \TT} \trace\left((\G_t^H\G_t)^{-1}\right).
    \end{align}
    \item Baseline Without AAI Mitigation:
    \begin{align}
        \nonumber
        & \MSE(\hat\x_{\text{n}}) = \MSE(\hat\x_{\text{g}}) + \dfrac{1}{TK\tau} \h^H \Bigg( \sum_{t\in\TT}
        \\
        & \qquad \V_t^* \SSSS_t^* \SSSS_t^T \V_t^T \otimes \G_t (\G_t^H \G_t)^{-2} \G_t^H \Bigg) \h.
    \end{align}
    \item Proposed Method for Full-Rank $\A$:
    \begin{align}
        \label{MSE_proposed}
        \nonumber
        & \MSE(\hat\x) = \MSE(\hat\x_{\text{g}}) + \dfrac{1}{TK\tau} \sum_{t\in\TT} \trace \Big(
        \\
        & \qquad ( \V_t^* \SSSS_t^* \SSSS_t^T \V_t^T \otimes \G_t (\G_t^H \G_t)^{-2} \G_t^H) \E \Big).
    \end{align}
    \item Proposed Methods for Rank-Deficient $\A$:
    \begin{align}
        \label{MSE_fallback}
        \nonumber
        & \MSE(\hat{\mybar\x}) = \MSE(\hat\x_{\text{g}})|_{T=1} + \dfrac{1}{K(\tau-J)} \trace \Big(
        \\
        & \qquad ( \mybar\SSSS^* \mybar\SSSS^T \otimes \G (\G^H \G)^{-2} \G^H) \mybar\E \Big).
    \end{align}
\end{itemize}
We observe that the last three MSEs are expressed in terms of the genie-aided baseline MSE plus an additional term arising due to imperfect AP-AP channel estimation. Further, the MSEs of our two proposed methods are independent of the DL SNR $\rho_\text{d}$, which follows directly from the fact that the corresponding covariance matrices are independent of $\rho_\text{d}$ (see Section \ref{sec:indep}).

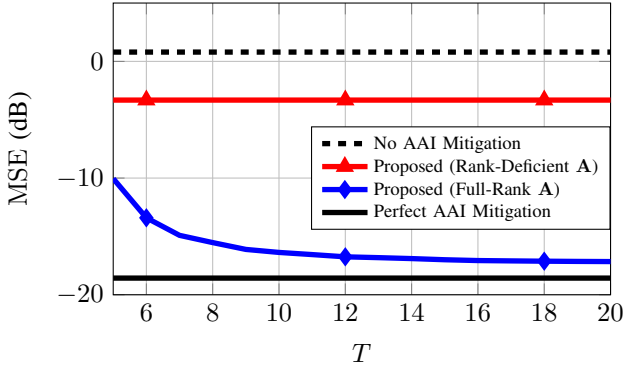
\begin{figure}[!t]
    \centering
    \begin{tikzpicture}
         \begin{axis}[
            width=0.45\textwidth,
            height=0.3\textwidth,
            xmajorgrids,
            ymajorgrids,
            xlabel=$T$,
            ylabel=$\MSE$ (dB),
            xmin=5,
            xmax=20,
            ymin=-20, 
            ymax=5,
            legend columns=1,
            legend style = {at={(1,0.4)},
            anchor={east},
            legend cell align=left,
            nodes={scale=0.7, transform shape}}]

            \addplot [color=black, dashed, line width=2.0pt]
              table[row sep=crcr]{%
            5	0.7969\\
            6	0.7969\\
            7	0.7969\\
            8	0.7969\\
            9	0.7969\\
            10	0.7969\\
            11 0.7969\\
            12 0.7969\\
            13 0.7969\\
            14 0.7969\\
            15 0.7969\\
            16 0.7969\\
            17 0.7969\\
            18 0.7969\\
            19 0.7969\\
            20 0.7969\\
            };
            \addlegendentry{No AAI Mitigation}            

            \addplot [color=red, line width=2.0pt, mark=triangle, mark options={solid}, mark phase=2, mark repeat=6]
              table[row sep=crcr]{%
            5	-3.3171\\
            6	-3.3171\\
            7	-3.3171\\
            8	-3.3171\\
            9	-3.3171\\
            10	-3.3171\\
            11 -3.3171\\
            12 -3.3171\\
            13 -3.3171\\
            14 -3.3171\\
            15 -3.3171\\
            16 -3.3171\\
            17 -3.3171\\
            18 -3.3171\\
            19 -3.3171\\
            20 -3.3171\\
            };
            \addlegendentry{Proposed (Rank-Deficient $\A$)}   
  
            \addplot [color=blue, line width=2.0pt, mark=diamond, mark options={solid, blue}, mark phase=2, mark repeat=6]
              table[row sep=crcr]{%
            5	-10.0476\\
            6	-13.4070\\
            7	-14.9090\\
            8	-15.5299\\ 
            9	-16.1149\\ 
            10	-16.3733\\
            11 -16.5578\\
            12 -16.7500\\
            13 -16.8204\\
            14 -16.9005\\
            15 -17.0001\\
            16 -17.0698\\
            17 -17.0988\\
            18 -17.1243\\
            19 -17.1433\\
            20 -17.1566\\
            };
            \addlegendentry{Proposed (Full-Rank $\A$)}
            
            \addplot [color=black, line width=2.0pt]
              table[row sep=crcr]{%
            5   -18.5727\\
            6	-18.5727\\
            7	-18.5727\\
            8	-18.5727\\
            9	-18.5727\\
            10  -18.5727\\
            11 -18.5727\\
            12 -18.5727\\
            13 -18.5727\\
            14 -18.5727\\
            15 -18.5727\\
            16 -18.5727\\
            17 -18.5727\\
            18 -18.5727\\
            19 -18.5727\\
            20 -18.5727\\
            };
            \addlegendentry{Perfect AAI Mitigation}
        \end{axis}
    \end{tikzpicture}
    \caption{
    MSEs of the proposed methods and the baselines for different $T$.
    }
    \label{Fig:MSE_var_T_MNKJ}
\end{figure}

\subsubsection{Numerical Evaluations}

We evaluate our proposed methods, both for full-rank and rank-deficient $\A$, by comparing their MSEs to the genie-aided baseline with perfect AAI mitigation and the baseline without AAI mitigation. Recall that both algorithms for the rank-deficient case are equal; they are not differentiated hereafter. Fig. \ref{Fig:MSE_var_T_MNKJ} presents the median MSEs obtained with each of the methods for varying $T \in \{ 5, \dots, 20 \}$. The medians are taken over 1000 independent UE deployments. We summarize our main findings:
\begin{itemize}
	\item We observe that the MSE of our proposed method for the full-rank case approaches that of the genie-aided baseline when $T$ grows. This indicates that the slow time-variations of the AP-AP channel can be effectively leveraged to enhance AAI mitigation in dynamic TDD. The intuition behind this behavior is that when observing the same interference channel over multiple UE-AP coherence intervals, the corresponding channel estimate will improve, making the AAI mitigation more efficient. For example, at $T=5$ the MSE of our proposed method exceeds that of the genie-aided baseline by roughly $9$ dB, whereas for $T>10$ this gap is smaller than $2$ dB.
	\item Our proposed method for the full-rank case heavily outperforms the baseline without AAI mitigation, even for small values of $T$. Already with $T=5$ the difference is larger than $10$ dB. With $T=20$, the same difference has grown to $18$ dB. 
	\item Our proposed method for the rank-deficient case gives an improvement of approximately $4$ dB over the baseline without AAI mitigation. This observation indicates that our joint estimation approach in itself achieves a performance gain, but that it is vital to leverage the slow time-variation of the AP-AP channel to reach a performance similar to the genie-aided baseline.
\end{itemize}

\begin{figure}[!t]
    \centering
    \begin{tikzpicture}
	\begin{groupplot}[
		group style={
		group size=1 by 2,
            vertical sep=2cm,
		},
            width=0.45\textwidth,
            height=0.3\textwidth,
            xmajorgrids,
            ymajorgrids,
            legend columns=1,
            transpose legend,
            legend style = {at={(1,1)},
            anchor={north east},
            legend cell align=left,
            nodes={scale=0.7, transform shape}}
            ]
            \nextgroupplot[title={(a) MSEs for varying UL SNRs.},xlabel={UL SNR (dB)},ylabel={$\MSE$ (dB)}, xmin=0, xmax=20, ymin=-25, ymax=25]

            \addplot [color=black, dashed, line width=2.0pt]
              table[row sep=crcr]{%
            0	10.8517\\
            5	5.6591\\
            10	0.7969\\
            15	-4.2450\\
            20	-9.1929\\
            };
            \addlegendentry{No AAI Mitigation}
            
            \addplot [color=red, line width=2.0pt, mark=triangle, mark options={solid, red}, mark phase=2, mark repeat=2]
              table[row sep=crcr]{%
            0	6.7258\\
            5	1.6467\\
            10	-3.3171\\
            15	-8.4337\\
            20	-13.4200\\
            };
            \addlegendentry{Proposed (Rank-Deficient $\A$)}   
  
           \addplot [color=blue, line width=2.0pt, mark=square*, mark options={solid, blue}, mark phase=2, mark repeat=2]
              table[row sep=crcr]{%
            0	-0.0378\\
            5	-5.0568\\
            10	-10.0476\\
            15	-15.0557\\
            20	-20.0024\\
            };                                  
            \addlegendentry{Proposed (Full-Rank $\A$, $T=5$)}

            \addplot [color=blue, line width=2.0pt, mark=*, mark options={solid, blue}, mark phase=2, mark repeat=2]
              table[row sep=crcr]{%
            0	-7.1468\\
            5	-12.1658\\
            10	-17.1566\\
            15	-22.1647\\
            20	-27.1114\\
            };                                  
            \addlegendentry{Proposed (Full-Rank $\A$, $T=20$)}
            
            \addplot [color=black, line width=2.0pt]
              table[row sep=crcr]{%
            0	-8.5629\\
            5	-13.5819\\
            10	-18.5727\\
            15	-23.5808\\
            20	-28.5275\\
            };
            \addlegendentry{Perfect AAI Mitigation}
            
            \nextgroupplot[title={(b) MSEs for varying DL SNRs.},xlabel={DL SNR (dB)},ylabel={$\MSE$ (dB)}, xmin=0, xmax=30, ymin=-25, ymax=25,legend style={at={(0,1)}, anchor=north west}]
            \addplot [color=black, dashed, line width=2.0pt]
              table[row sep=crcr]{%
            0	-15.5853\\
            5   -12.7709\\
            10	-8.6513\\
            15	-3.9585\\
            20	0.7969\\
            25	5.8468\\
            30  10.7783\\
            };
            \addlegendentry{No AAI Mitigation}            

            \addplot [color=red, line width=2.0pt, mark=triangle, mark options={solid, red}, mark phase=2, mark repeat=4]
                table[row sep=crcr]{%
            0	-3.3171\\
            5	-3.3171\\
            10	-3.3171\\
            15	-3.3171\\
            20	-3.3171\\
            25	-3.3171\\
            30  -3.3171\\
            };
            \addlegendentry{Proposed (Rank-Deficient $\A$)}s
            
           \addplot [color=blue, line width=2.0pt, mark=square*, mark options={solid, blue}, mark phase=2, mark repeat=4]
              table[row sep=crcr]{%
            0	-10.0476\\
            5	-10.0476\\
            10	-10.0476\\
            15	-10.0476\\
            20	-10.0476\\
            25	-10.0476\\
            30  -10.0476\\
            };                                  
            \addlegendentry{Proposed (Full-Rank $\A$, $T=5$)}

            \addplot [color=blue, line width=2.0pt, mark=*, mark options={solid, blue}, mark phase=2, mark repeat=4]
              table[row sep=crcr]{%
            0	-17.1566\\
            5	-17.1566\\
            10	-17.1566\\
            15	-17.1566\\
            20	-17.1566\\
            25	-17.1566\\
            30  -17.1566\\
            };                                  
            \addlegendentry{Proposed (Full-Rank $\A$, $T=20$)}

            \addplot [color=black, line width=2.0pt]
              table[row sep=crcr]{%
            0	-18.5727\\
            5	-18.5727\\
            10	-18.5727\\
            15	-18.5727\\
            20	-18.5727\\
            25	-18.5727\\
            30  -18.5727\\
            };
            \addlegendentry{Perfect AAI Mitigation}

        \end{groupplot}
    \end{tikzpicture}
    
    \caption{
    Comparison between the proposed methods and the baselines with different values of the UL and DL SNRs.
    }
    \label{Fig:MSE_var_UL_DL}
\end{figure}

Next, in Fig. \ref{Fig:MSE_var_UL_DL}, we present the median MSEs obtained with the proposed methods and baselines for varying UL and DL SNRs $\rho_{\text{u}}$ and $\rho_{\text{d}}$. We make the following observations:
\begin{itemize}
	\item When varying the UL SNR, all MSEs decrease linearly (in log-log scale) but the gaps remain equal. This behavior is expected as the UL SNR enters the MSE expressions through $\{ \G_t \}$. It can be seen that the genie MSE and all terms arising due to the imperfect AP-AP channel estimation are inversely proportional to the UL SNR.
	\item As mentioned in Section \ref{sec:MSE} above, the MSEs of our proposed methods are independent of the DL SNR. This behavior can also be observed from Fig.~\ref{Fig:MSE_var_UL_DL}(b).
	\item The baseline without AAI mitigation performs better than the proposed method for the rank-deficient case at low DL SNRs. This is because, for decreasing DL SNRs, the gain in mitigating the AAI becomes more limited. We also observe that the MSE without AAI mitigation approaches that of the genie-aided baseline at low DL SNRs, which is reasonable as there is virtually no gain in mitigating the AAI when its power is negligibly small. 
\end{itemize}

\subsection{Evaluation of Spectral Efficiencies}
\begin{figure}[!t]
    \centering
    \begin{tikzpicture}
         \begin{axis}[
            width=0.45\textwidth,
            height=0.3\textwidth,
            xmajorgrids,
            ymajorgrids,
            xlabel=$T$,
            ylabel=Median SE (bits/s/Hz),
            xmin=5,
            xmax=15,
            ymin=0,
            ymax=8,
            legend columns=1,
            legend style = {at={(1,0.22)},
            anchor={south east},
            legend cell align=left,
            nodes={scale=0.7, transform shape}}]

            \addplot [color=black, line width=2.0pt]
              table[row sep=crcr]{%
            5   6.1965\\
            6	6.1965\\
            7	6.1965\\
            8	6.1965\\
            9	6.1965\\
            10 6.1965\\
            11 6.1965\\
            12 6.1965\\
            13 6.1965\\
            14 6.1965\\
            15 6.1965\\
            };
            \addlegendentry{Perfect AAI Mitigation}

            \addplot [color=blue, line width=2.0pt, mark=diamond, mark options={solid, blue}, mark phase=2, mark repeat=6]
              table[row sep=crcr]{%
            5	3.8442\\
            6	5.0380\\
            7	5.3585\\
            8	5.5009\\ 
            9	5.6544\\
            10	5.7992\\
            11 5.8578\\
            12 5.8731\\
            13 5.8952\\
            14 5.9064\\
            15 5.9269\\
            };
            \addlegendentry{Proposed (Full-Rank $\A$)}
    
            \addplot [color=red, line width=2.0pt, mark=triangle, mark options={solid}, mark phase=2, mark repeat=6]
              table[row sep=crcr]{%
            5	1.4172\\
            6	1.4172\\
            7	1.4172\\
            8	1.4172\\
            9	1.4172\\
            10	1.4172\\
            11 1.4172\\
            12 1.4172\\
            13 1.4172\\
            14 1.4172\\
            15 1.4172\\
            };
            \addlegendentry{Proposed (Rank-Deficient $\A$)}   

            \addplot [color=black, dashed, line width=2.0pt]
              table[row sep=crcr]{%
            5	0.8572\\
            6	0.8572\\
            7	0.8572\\
            8	0.8572\\
            9	0.8572\\
            10	0.8572\\
            11 0.8572\\
            12 0.8572\\
            13 0.8572\\
            14 0.8572\\
            15 0.8572\\
            };
            \addlegendentry{No AAI Mitigation}

            \addplot [color=green!70!black, line width=2.0pt, mark=square, mark options={solid}, mark phase=2, mark repeat=6]
              table[row sep=crcr]{%
            5	0.5064\\
            6	0.5064\\
            7	0.5064\\
            8	0.5064\\ 
            9	0.5064\\
            10	0.5064\\
            11 0.5064\\
            12 0.5064\\
            13 0.5064\\
            14 0.5064\\
            15 0.5064\\
            };
            \addlegendentry{Zero-Forcing Combining}

        \end{axis}
    \end{tikzpicture}
    \caption{
    SEs of the proposed methods and the baselines for different $T$.
    }
    \label{Fig:SE_var_T}
\end{figure}
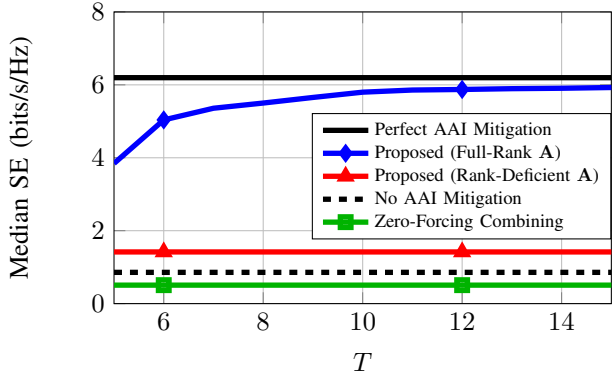

Next, we derive achievable UL SEs for our proposed methods and benchmark these against an existing AAI mitigation method with zero-forcing combining. The genie-aided method and the method without AAI mitigation will serve as baselines.

\subsubsection{Proposed Method for Full-Rank $\A$}
\label{SE_proposed}

The least-squares solution in \eqref{lszhat} can be written as
\begin{align}
    \hat\z = \left( \A^H \A \right)^{-1} \A^H \y = 
    \begin{bmatrix}
        \x
        \\
        \h
    \end{bmatrix}
    + \left( \A^H \A \right)^{-1} \A^H \w.
\end{align}
Thus, for any symbol $s \in \{1,\dots,TK\tau\}$ of $\x$, the estimate is
\begin{align}
    \hat{x}_s = x_s + \aaa_s^H \w,
\end{align}
where $\aaa_s^H$ is the $s$th row of $( \A^H \A )^{-1} \A^H$. Note that all $\{ \aaa_s \}$ depend on the matrices $\{ \G_t \}$, $\{ \V_t \}$, and $\{ \SSSS_t \}$, which are known at the UL AP. Hence, we can treat these matrices as side information and use \cite[Eq. (2.48)]{ngo16} to obtain the following achievable UL SE:
\begin{align}
    \nonumber
    & \EEE \Bigg\{ \log_2 \left( 1 + \dfrac{1}{\var \{ \aaa_s^H \w | \{ \G_t \}, \{ \SSSS_t \}, \{ \V_t \} \}} \right) \Bigg\}
    \\
    = & \EEE \Bigg\{ \log_2 \left( 1 + \dfrac{1}{\| \aaa_s \|^2} \right) \Bigg\},
\end{align}
where the expectation is taken over $\{ \G_t \}$, $\{ \SSSS_t \}$, and $\{ \V_t \}$.

\subsubsection{Proposed Method for Rank-Deficient $\A$}

From \eqref{zBarHat}, the data symbol estimate for any $s \in \{ 1,\dots,K(\tau-J) \}$ is
\begin{align}
    \hat{\mybar{x}}_s = \mybar{x}_s + \mybar{\aaa}_s^H \w,
\end{align}
where $\mybar{\aaa}_s^H$ is row $s$ of $( \mybar\A^H \mybar\A )^{-1} \mybar\A^H$. By following the same reasoning as in Section \ref{SE_proposed}, an achievable UL SE is
\begin{align}
     \dfrac{\tau-J}{\tau} \EEE \Bigg\{ \log_2 \left( 1 + \dfrac{1}{\| \mybar{\aaa}_s \|^2} \right) \Bigg\},
\end{align}
where the pre-log factor accounts for the UL UEs being silent in $J$ samples of each UE-AP coherence interval. The expectation is taken over $\G$ and $\SSSS$.

\subsubsection{Genie-Aided Baseline}

For any $s \in \{ 1,\dots,TK\tau \}$, the data symbol estimate from \eqref{eq:xg} is
\begin{align}
    \hat{x}_{\text{g},s} = x_s + \dot{\aaa}_s^H \w,
\end{align}
where $\dot{\aaa}_s^H$ is row $s$ of $( \dot\A^H \dot\A )^{-1} \dot\A^H$. An achievable SE is
\begin{align}
     \EEE \Bigg\{ \log_2 \left( 1 + \dfrac{1}{\| \dot{\aaa}_s \|^2} \right) \Bigg\},
\end{align}
where the expectation is taken over $\{ \G_t \}$.

\subsubsection{Baseline Without AAI Mitigation}

From \eqref{eq:xhatn}, the data symbol estimate for any $s \in \{ 1,\dots,TK\tau \}$ is
\begin{align}
    \hat{x}_{\text{n},s} = x_s + \bb^H_s \h + \dot{\aaa}_s^H \w,
\end{align}
where $\bb_s^H$ is the $s$th row of the matrix $\B$ defined in \eqref{eq:xhatn}. The conditional variance of the effective noise term is
\begin{align}
    \var \{ \bb_s^H \h + \dot{\aaa}_s^H \w | \{ \G_t \}, \{ \SSSS_t \}, \{ \V_t \} \} = \alpha \| \bb_s \|^2 + \| \dot{\aaa}_s \|^2.
\end{align}
With the same reasoning as before, an achievable SE is
\begin{align}
     \EEE \Bigg\{ \log_2 \left( 1 + \dfrac{1}{\alpha \| \bb_s \|^2 + \| \dot{\aaa}_s \|^2} \right) \Bigg\},
\end{align}
where the expectation is taken over $\{ \G_t \}$, $\{ \SSSS_t \}$, and $\{ \V_t \}$.

\subsubsection{Zero-Forcing Combining for AAI Mitigation}

As a benchmark we consider zero-forcing combining at the UL AP as in \cite{daSilva2021WC}, which is the existing method closest to ours. This method requires AP-AP channel knowledge and hence an overhead for pilot transmission between the APs. We use $M$ samples to transmit orthogonal pilot sequences from the DL AP to the UL AP, and apply standard minimum mean-square error channel estimation to obtain an estimate $\hat\HH$ of $\HH$ \cite[Sec. 3.1]{ngo16}. Define the estimation error $\tilde\HH = \HH - \hat\HH$, and write the received signal \eqref{Yt} as (we drop the index $t$ here since we work in a single coherence interval) 
\begin{align}
    \Y =
    \begin{bmatrix}
        \G & \hat\HH \V
    \end{bmatrix}
    \begin{bmatrix}
        \X
        \\
        \SSSS
    \end{bmatrix}
    + \tilde\HH \V \SSSS + \W.
\end{align}
Then, the zero-forcing combining matrix is
\begin{align}
    \Z = \left(
    \begin{bmatrix}
        \G & \hat\HH \V
    \end{bmatrix}^H
    \begin{bmatrix}
        \G & \hat\HH \V
    \end{bmatrix}
    \right)^{-1}
    \begin{bmatrix}
        \G & \hat\HH \V
    \end{bmatrix}^H,
\end{align}
which gives
\begin{align}
    \label{eq:zf}
    \Z \Y = 
    \begin{bmatrix}
        \X
        \\
        \SSSS
    \end{bmatrix}
    + \underbrace{\Z (\tilde\HH \V \SSSS + \W)}_{\triangleq \N}.
\end{align}
Let $x_s$, $s \in \{ 1, \dots, K\tau \}$, be any data sample of $\X$ in \eqref{eq:zf}. Then, the estimate of $x_s$ can be expressed as
\begin{align}
    \hat{x}_s = x_s + n_s,
\end{align}
where $n_s$ is the $s$th element of the matrix $\N$ defined in \eqref{eq:zf}. An achievable SE is given by
\begin{align}
    \label{eq:SE_ZF}
    \dfrac{\tau - M/T}{\tau} \EEE \Bigg\{ \log_2 \left( 1 + \dfrac{1}{\var \{ n_s | \G , \SSSS , \V \}} \right) \Bigg\},
\end{align}
where the expectation is taken over $ \G $, $ \SSSS $, and $\V$. 
Note that we extend \cite{daSilva2021WC} by leveraging the slow time-variation of the AP-AP channel, such that the $M$ pilot samples can be spread over $T$ UE-AP coherence intervals. This reduces the pilot overhead to $M/T$ pilot samples per UE-AP coherence interval (on average), which is reflected in the pre-log factor of \eqref{eq:SE_ZF}. 

\subsubsection{Numerical Evaluation}

In Fig. \ref{Fig:SE_var_T}, we evaluate the median average SE obtained by our proposed methods, the two baselines, and the zero-forcing combining technique, for $T \in \{ 5, \dots, 15 \}$. The average is taken over the $K$ UL UEs and the median is taken over 100 independent UE deployments. We make the following observations:
\begin{itemize}
    \item The behavior of the SE is similar to the behavior of the MSE, both for the proposed methods and for the two baselines. Our proposed method for the full-rank case heavily outperforms the proposed method for the rank-deficient case and the baseline without AAI mitigation, and its SE approaches that of the genie-aided baseline when $T$ grows.
    \item In terms of SE, the proposed method for the rank-deficient case still outperforms the baseline without AAI mitigation, even though it suffers from a reduced pre-log factor due to the UL UEs being silent during $J$ samples.
    \item The zero-forcing combining benchmark performs worse than the baseline without AAI mitigation. This is because the zero-forcing combining matrix requires $M \geq K + J$ to be able to separate the $K$ UL UEs and cancel the $J$ AAI streams of the effective AP-AP channel $\HH \V$. Thus, with $M = K + J$ as in Fig. \ref{Fig:SE_var_T}, too many degrees of freedom are spent on AAI mitigation to allow for accurate data detection.
\end{itemize}
To obtain a fair comparison between our proposed methods and the zero-forcing benchmark, we now evaluate the SE for a varying number of AP antennas $M = N \in \{ 16, \dots, 32 \}$ for a fixed number of $K = J = 8$ UL/DL UEs. In Fig. \ref{Fig:SE_var_MN}, we compare the zero-forcing benchmark with $T=1$ and $T=10$ to our proposed method for the full-rank case with $T=10$ and to our proposed method for the rank-deficient case.\footnote{Note that our theoretical analysis does not guarantee full-rank $\A$ here; however, numerical evaluations show that with $T=10$, $\A$ has full rank for all setups considered in Fig. \ref{Fig:SE_var_MN}.} The two baselines are also included. We observe the following:
\begin{itemize}
    \item With $T=10$, our proposed method performs close to the genie-aided baseline, and outperforms the benchmark with zero-forcing combining. This is especially prominent when the number of AP antennas is small, since our proposed method (unlike the zero-forcing benchmark) does not sacrifice spatial dimensions for AAI mitigation, and hence does not require the number of UL AP antennas to be much larger than the number of UEs.
    \item Our proposed method for the rank-deficient case outperforms the zero-forcing benchmark with $T=1$ when the number of AP antennas is either small or large. When increasing the number of AP antennas further, the performance of the zero-forcing benchmark will continue to decrease due to the increased channel estimation overhead which reduces the pre-log factor.
    \item Except for when the number of AP antennas is smallest (where $M = K+J$ leaves too few spatial dimensions for zero-forcing), the zero-forcing benchmark with $T=10$ outperforms our proposed method for the rank-deficient case. This indicates that zero-forcing combining might be a better fallback option in scenarios where $T>1$, but not sufficiently large to make the proposed method for the full-rank case applicable.
\end{itemize}

In conclusion, our proposed methods achieve gains over the zero-forcing benchmark in terms of the UL SE in most cases. We end with a remark on the DL performance.
\begin{remark}
    We restrict the simulations to the UL SE. However, our proposed method also improves the DL performance over existing methods, since we overcome the AP-AP channel estimation overhead which increases the pre-log factor of the DL SE. Our proposed method for the full-rank case increases the pre-log factor of the DL SE by $M/T$ samples per UE-AP coherence interval, whereas the DL SINR is the same.
\end{remark}

\begin{figure}[!t]
    \centering
    \begin{tikzpicture}
         \begin{axis}[
            width=0.45\textwidth,
            height=0.3\textwidth,
            xmajorgrids,
            ymajorgrids,
            xlabel={$M=N$},
            ylabel=Median SE (bits/s/Hz),
            xmin=16,
            xmax=32,
            ymin=0,
            ymax=8,
            legend columns=1,
            legend style = {at={(1,1)},
            anchor={south east},
            legend cell align=left,
            nodes={scale=0.7, transform shape}}]

            \addplot [color=black, line width=2.0pt]
              table[row sep=crcr]{%
            16	5.6885\\
            20	6.5295\\
            24	6.7988\\
            28	7.1318\\
            32	7.4114\\
            };
            \addlegendentry{Perfect AAI Mitigation}

            \addplot [color=blue, line width=2.0pt, mark=diamond, mark options={solid, blue}, mark phase=2, mark repeat=2]
              table[row sep=crcr]{%
            16	5.5410\\
            20  6.2033\\
            24	6.5156\\
            28	6.7553\\ 
            32	6.9182\\
            };
            
            \addlegendentry{Proposed (Full-Rank $\A$, $T=10$)}
            
            \addplot [color=green!70!black, line width=2.0pt, mark=o, mark options={solid}, mark phase=3, mark repeat=4]
              table[row sep=crcr]{%
            16	1.1487\\
            20	4.5233\\
            24	5.5735\\
            28	5.9259\\
            32	6.2504\\
            };
            \addlegendentry{Zero-Forcing Combining ($T=10$)}
            
            \addplot [color=red, line width=2.0pt, mark=triangle, mark options={solid}, mark phase=2, mark repeat=2]
              table[row sep=crcr]{%
            16	2.0539\\
            20	2.4138\\
            24	2.7870\\
            28	2.9142\\ 
            32	3.0773\\
            };
            \addlegendentry{Proposed (Rank-Deficient $\A$)}   

            \addplot [color=green!70!black, line width=2.0pt, mark=square, mark options={solid}, mark phase=3, mark repeat=4]
              table[row sep=crcr]{%
            16	0.8070\\
            20	2.8271\\
            24	3.0444\\
            28	2.7621\\ 
            32	2.4040\\
            };
            \addlegendentry{Zero-Forcing Combining ($T=1$)}

            \addplot [color=black, dashed, line width=2.0pt]
              table[row sep=crcr]{%
            16	0.9040\\
            20	1.1231\\
            24	1.3965\\
            28	1.5796\\ 
            32	1.6746\\
            };
            \addlegendentry{No AAI Mitigation}

        \end{axis}
    \end{tikzpicture}
    \caption{
    SEs of the proposed methods and the baselines for different $M=N$.
    }
    \label{Fig:SE_var_MN}
\end{figure}
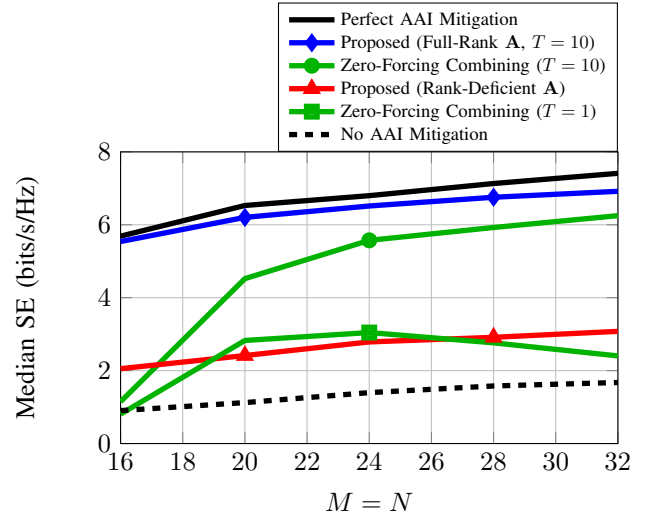

\vspace{-3mm}
\section{Conclusions and Future Directions}

We conclude that slowly time-varying AP-AP channels can be effectively leveraged to enhance interference mitigation in dynamic TDD. To this end, we have proposed a method for AP-AP interference mitigation that uses a least-squares formulation to jointly estimate the UL data and the AP-AP interference channel over multiple coherence intervals, during which the interference channel stays constant. We observed that this problem does not always have a unique solution, and provided an analytical bound on how slowly the AP-AP channel must vary to make the problem uniquely solvable. We solved the formulated least-squares problem in the full-rank case, and proposed a modification of the original problem that has a unique solution in the rank-deficient case.

Numerical simulations showed that our proposed method for the full-rank case significantly outperforms a baseline without interference mitigation and a benchmark zero-forcing receiver, in all relevant cases. Further, it achieves a similar performance as a genie-aided baseline with perfect interference mitigation for very slowly time-varying AP-AP channels.

We must leave the problem of calculating the exact rank of the regression matrix in the general scenario in Theorem \ref{theorem:dimNA} as an open question. Other directions for future work could be to develop improved methods for the rank-deficient case, and to extend our proposed method for the full-rank case to overcome the strict block-fading assumption used herein. A further direction is to extend the batch least-squares solution to an online recursive variant that decodes each coherence interval as it arrives, thereby reducing the latency caused by jointly processing multiple coherence intervals.

\appendices

\renewcommand{\thesectiondis}[2]{\Alph{section}:}
\section{Proof of Theorem \ref{theorem:dimNA}}
\label{appendix:dimNA}

We will split the proof into two parts. The first part proves \eqref{eq:rankA}, and the second part derives \eqref{eq:boundT} by building upon the results from part one. The proof of \eqref{eq:rankA} can be simplified by considering the leftmost $TM\tau \times TK\tau$ block of $\A$, which clearly has full-rank. However, we provide a proof including some additional details that are not needed to prove \eqref{eq:rankA}, but necessary later to derive \eqref{eq:boundT}.

\subsubsection{Proof of \eqref{eq:rankA}}

We define the matrix $\PP_t = \SSSS_t^T \V_t^T$. For notational convenience, we break the convention that bold lowercase symbols are column vectors. Instead, denote by $\p_t^{(i)} \in \C^{1 \times N},\ i \in \{ 1,\dots,\tau \},$ row $i$ of $\PP_t$, and let $p_t^{(i,j)},\ i \in \{ 1,\dots,\tau \},\ j \in \{ 1,\dots,N \},$ be element $(i,j)$ of $\PP_t$ or equivalently element $j$ of $\p_t^{(i)}$. Then we have
\begin{align}
    \label{Abig}
    \A =
    \begin{bmatrix}
        \G_1 & & & & \0 & \p_1^{(1)} \otimes \II_M
        \\
        & \ddots & & & & \vdots
        \\
        & & \G_t & & & \p_t^{(i)} \otimes \II_M
        \\
        & & & \ddots & & \vdots
        \\
        \0 & & & & \G_T & \p_T^{(\tau)} \otimes \II_M
    \end{bmatrix}.
\end{align}
The matrix $\A$ consists of $T\tau$ row blocks each with $M$ rows. By excluding all columns consisting only of zeros, each block can be written as
\begin{align}
    \label{BlockWise}
    \begin{bmatrix}
        \G_t & \p_t^{(i)} \otimes \II_M
    \end{bmatrix}
    =
    \begin{bmatrix}
        \G_t & p_t^{(i,1)} \II_M & \dots & p_t^{(i,N)} \II_M
    \end{bmatrix}.
\end{align}
We now perform row operations on each block independently. The columns with only zeros are not affected by row operations; thus we can work with the structure in \eqref{BlockWise}. From Assumption \ref{assum:VS}, it is possible to perform a sequence of row operations on $\G_t$ to obtain a similar matrix $\check\G_t$ on row echelon form with pivot elements in the first $K$ rows:
\begin{align}
    \label{RowG}
    \G_t \overset{\text{row ops.}}{\sim}
    \check\G_t \triangleq
    \;\begin{bNiceMatrix}[name=B,last-col=4]
        \times & \dots & \times &  \\
        0 & \ddots & \vdots & \quad K \text{ rows} \\
        \vdots & \ddots & \times & \\[2mm]
        0 & \dots & 0 & \\[-1mm]
        \vdots & \ddots & \vdots & \quad M-K \text{ rows} \\
        0 & \dots & 0 & \\
        \CodeAfter
        \SubMatrix.{1-3}{3-3}\}[extra-height=0mm,xshift=1mm]
        \SubMatrix.{4-3}{6-3}\}[extra-height=0mm,xshift=1mm]
    \end{bNiceMatrix}\;
\end{align}
Here, the symbol "$\times$" denotes an arbitrary non-zero element. When performing these row operations on the full block in \eqref{BlockWise}, the remaining part $\p_t^{(i)} \otimes \II_M$ is   affected in the following way:
\begin{align}
    \nonumber
    \p_t^{(i)} \otimes \II_M & =
    \begin{bmatrix}
        p_t^{(i,1)} \II_M & \dots & p_t^{(i,N)} \II_M
    \end{bmatrix}
    \\
    & \overset{\text{row ops.}}{\sim} 
    \begin{bmatrix}
        p_t^{(i,1)} \check\Q_t & \dots & p_t^{(i,N)} \check\Q_t
    \end{bmatrix}		
    = \p_t^{(i)} \otimes \check\Q_t,
\end{align}
where $\check\Q_t \in \C^{M \times M}$ is the following lower triangular matrix:
\begin{align}
    \check\Q_t =
    \;\begin{bNiceMatrix}[name=BB]
        1 & 0 & \dots & \dots & \dots & \dots & 0 \\
        \times & \ddots & \ddots & & & & \vdots \\
        \vdots & \ddots & 1 & 0 & \dots & \dots & 0 \\
        \times  & \dots & \times & 1 & 0 & \dots & 0  \\
        \vdots & & \vdots & 0 & \ddots & \ddots & \vdots \\
        \vdots &  & \vdots & \vdots & \ddots & \ddots & 0 \\
        \times & \dots & \times & 0 & \dots & 0 & 1 \\
        \CodeAfter
        \UnderBrace[shorten, yshift=1mm]{7-4}{7-7}{M-K \text{ columns}}
        \UnderBrace[shorten, yshift=1mm]{7-1}{7-3}{K \text{ columns}}
    \end{bNiceMatrix}\;.
\end{align}
\vspace{+3mm}

As before, the elements "$\times$" denote arbitrary non-zero elements. Now, after performing the sequence of row operations that obtains \eqref{RowG}, the full block \eqref{BlockWise} has the following form:
\begin{align}
    \label{RowBlock}
    \begin{bmatrix}
        \G_t & \p_t^{(i)} \otimes \II_M
    \end{bmatrix}
    \overset{\text{row ops.}}{\sim}
    \begin{bmatrix}
        \check\G_t & \p_t^{(i)} \otimes \check\Q_t
    \end{bmatrix},
\end{align}
and the full matrix in \eqref{Abig} becomes
\begin{align}
    \label{RowA}
    \A \overset{\text{row ops.}}{\sim} \check\A \triangleq
    \begin{bmatrix}
        \II_{\tau}\otimes \check\G_1 & & \0 & \PP_1 \otimes \check\Q_{1}
        \\
        & \ddots & & \vdots
        \\
        \0 & & \II_{\tau}\otimes \check\G_T & \PP_T \otimes \check\Q_T
    \end{bmatrix}.
\end{align}
We note that $\check\A$ has $T\tau$ blocks of the form \eqref{RowBlock}, each with $K$ pivot elements. Hence, there are $TK\tau$ pivot elements in total, and the rank of $\A$ is at least $TK\tau$. It remains to consider the submatrix with the remaining $M-K$ rows from each of the $T\tau$ row blocks and compute its rank. We collect these $T\tau(M-K)$ rows (excluding columns with zeros) in the matrix
\begin{align}
    \label{eq:RR}
    \RR \triangleq
    \begin{bmatrix}
        \PP_1 \otimes \Q_1
        \\
        \vdots
        \\
        \PP_T \otimes \Q_T
    \end{bmatrix}
     \in \C^{T\tau(M-K) \times MN},
\end{align}
where $\Q_t \in \C^{(M-K) \times M}$ is the bottom $M-K$ rows of $\check\Q_t$. Specifically, $\Q_t$  can be expressed as
\begin{align}
    \label{Q}
    \Q_t \triangleq
    \begin{bmatrix}
        \mybar\Q_t & \II_{M-K}
    \end{bmatrix},
\end{align}
where $\mybar\Q_t \in \C^{(M-K) \times K}$ is the bottom left block of $\check\Q_t$. Now, with $R \triangleq \rank\left( \RR \right)$, we can express the rank of $\A$ as
\begin{align}
    \label{eq:rankA_R}
    \rank\left( \A \right) = TK\tau + R,
\end{align}
with $R \leq MN$, and the first part of the proof is completed.

\subsubsection{Proof of \eqref{eq:boundT}}
Now, take a set of row vectors spanning the row space of $\PP_t$, say its first $J$ rows, and collect these in
\begin{align}
    \mybar\PP_t \triangleq
    \begin{bmatrix}
        \p_t^{(1)}
        \\
        \vdots
        \\
        \p_t^{(J)}
    \end{bmatrix}
    \in \C^{J \times N}.
\end{align}
Then, the row space of $\PP_t \otimes \Q_t$ is spanned by the rows of $\mybar\PP_t \otimes \Q_t$. Thus, $R = \rank(\RR) = \rank(\check\RR)$, where
\begin{align}
    \label{eq:checkRR}
    \check\RR \triangleq
    \begin{bmatrix}
        \mybar\PP_1 \otimes \Q_1
        \\
        \vdots
        \\
        \mybar\PP_T \otimes \Q_T
    \end{bmatrix}
    \in \C^{T(M-K)J \times MN}.
\end{align}
It follows that
\begin{align}
    \label{eq:R_Lemma}
    R = \rank(\check\RR) \leq \min \{ T(M-K)J,\ MN \}.
\end{align}
From \eqref{eq:dimNA} and \eqref{eq:R_Lemma} we get
\begin{align}
    \nonumber
    \dim(\NN(\A)) & = MN - R
    \\
    & \geq MN - \min \{ T(M-K)J,\ MN \}.
\end{align}
Hence, the dimension of the nullspace can only equal zero if $\check\RR$ is wide, i.e., if $T(M-K)J \geq MN$, which gives \eqref{eq:boundT}.

\vspace{-3mm}
\renewcommand{\thesectiondis}[2]{\Alph{section}:}
\section{Proof of Theorem \ref{theorem:M=N=2K=2J}}
\label{appendix:SC_MN_2KJ}

With $M=N=2K=2J$, $\A$ and $\check\RR$ have the following dimensions (expressed in $K$):
\begin{align}
    & \A \in \C^{2TK\tau \times (TK\tau + 4K^2)}
    \\
    \label{eq:barR_Case2}
    & \check\RR \in \C^{TK^2 \times 4K^2}.
\end{align}
By inserting \eqref{Q} into $\check\RR$ in \eqref{eq:checkRR} and permuting its columns we get the following matrix $\mybar\RR$ which is similar to $\check\RR$:
\begin{align}
    \label{eq:barR}
    \check\RR \overset{\text{col. ops.}}{\sim} \mybar\RR \triangleq
    \begin{bmatrix}
        \mybar\PP_1 \otimes \mybar\Q_1 & \mybar\PP_1 \otimes \II_{K}
        \\
        \vdots & \vdots
        \\
        \mybar\PP_T \otimes \mybar\Q_T & \mybar\PP_T \otimes \II_{K}
    \end{bmatrix}.
\end{align}
Thus,
\begin{align}
    \label{eq:R}
    R = \rank(\RR) = \rank(\check\RR) = \rank(\mybar\RR).
\end{align}
Since $\A$ is tall (or square) it has full rank if $\rank(\A) = TK\tau + 4K^2$. From \eqref{eq:rankA_R} and \eqref{eq:R}, this is equivalent to $R = \rank(\mybar\RR) = 4K^2$. Further, we have $\rank(\mybar\RR) \leq \min \{ TK^2,\ 4K^2 \}$, and it follows that $\rank(\mybar\RR) = 4K^2$ is only obtainable  $TK^2 \geq 4K^2 \iff T \geq 4$.

It remains to prove that $\rank(\mybar\RR) = 4K^2$ is always fulfilled if $T \geq 4$. To this end, it is sufficient to show that $\mybar\RR$ is full-rank for $T=4$. Then it follows that $\mybar\RR$ is full-rank for all $T \geq 4$; adding more rows to a matrix cannot decrease its rank. With $T=4$, $\mybar\RR$ has the following structure:
\begin{align}
    \label{eq:barR_T4}
    \mybar\RR =
    \begin{bmatrix}
        \mybar\PP_1 \otimes \mybar\Q_1 & \mybar\PP_1 \otimes \II_{K}
        \\
        \mybar\PP_2 \otimes \mybar\Q_2 & \mybar\PP_2 \otimes \II_{K}
        \\
        \mybar\PP_3 \otimes \mybar\Q_3 & \mybar\PP_3 \otimes \II_{K}
        \\
        \mybar\PP_4 \otimes \mybar\Q_4 & \mybar\PP_4 \otimes \II_{K}
    \end{bmatrix}
    \in \C^{4K^2 \times 4K^2}.
\end{align}
Since $\mybar\RR$ is a square matrix, it has full-rank if and only if its determinant is non-zero. We will prove that it is non-zero with probability one by showing that $\mybar\RR$ can be written as a polynomial in independent random variables. 
To this end, first consider $\det(\mybar\RR)$, which can be written as a sum of $(4K^2)!$ terms, where each term is a (negative or positive) product of $4K^2$ elements of $\mybar\RR$. Each term is formed such that there is only one element from each row and each column included in the product \cite[Sec. 0.3.2]{horn2012matrix}. This is a polynomial, say $p(\cdot)$, in the entries of $\{ \mybar\PP_t \}$ and $\{ \mybar\Q_t \}$, $t \in \{ 1, \dots, 4 \}$. Define
\begin{align}
    \PP =
    \begin{bmatrix}
        \mybar\PP_1
        \\
        \mybar\PP_2
        \\
        \mybar\PP_3
        \\
        \mybar\PP_4
    \end{bmatrix}
    ,\
    \Q =
    \begin{bmatrix}
        \mybar\Q_1
        &
        \mybar\Q_2
        &
        \mybar\Q_3
        &
        \mybar\Q_4
    \end{bmatrix}.
\end{align}
Then $p(\cdot)$ is a polynomial in the entries of these matrices, i.e., $\det(\mybar\RR) = p(\PP, \Q)$. This polynomial can be partitioned as
\begin{align}
    \label{eq:detRbar}
    p(\PP, \Q) & = 
    f(\PP, \Q)
    \det \left(
    \begin{bmatrix}
        \mybar\PP_3 \otimes \II_{K}
        \\
        \mybar\PP_4 \otimes \II_{K}
    \end{bmatrix}
    \right) + r(\PP, \Q),
\end{align}
with  
\begin{align}
    \label{eq:f}
    & f(\PP, \Q) = \prod_{k=1}^K \left( \det\left( \mybar{p}_1^{(k,k)} \mybar\Q_1 \right) \det\left(\mybar{p}_2^{(k,k+K)} \mybar\Q_2 \right) \right),
\end{align}
where $\mybar{p}_t^{(i,j)}$ is element $(i,j)$ of $\mybar\PP_t$. The remaining term $r(\PP, \Q)$ consists of the terms not captured by the first term of \eqref{eq:detRbar}. Note that $f(\PP, \Q)$ is not the zero polynomial (i.e., $f(\PP, \Q) \not\equiv 0$), since all factors in \eqref{eq:f} are non-zero. Furthermore, we have
\begin{align}
    \label{eq:detPI}
    \nonumber
    & \det \left(
    \begin{bmatrix}
        \mybar\PP_3 \otimes \II_{K}
        \\
        \mybar\PP_4 \otimes \II_{K}
    \end{bmatrix}
    \right)
    =
    \det \left(
    \begin{bmatrix}
        \mybar\PP_3
        \\
        \mybar\PP_4
    \end{bmatrix}
     \otimes \II_{K}
    \right)
    \\
    & =
    \det \left(
    \begin{bmatrix}
        \mybar\PP_3
        \\
        \mybar\PP_4
    \end{bmatrix}
    \right)^K
    \det ( \II_K )^{2K}
    =
    \det \left(
    \begin{bmatrix}
        \mybar\PP_3
        \\
        \mybar\PP_4
    \end{bmatrix}
    \right)^K
    \not\equiv 0,
\end{align}
where in the second equality we used a well-known property for the determinant of a Kronecker product \cite[Sec. 12.1]{lancaster1985theory}. Hence, the product of \eqref{eq:f} and \eqref{eq:detPI} in \eqref{eq:detRbar} is not the zero polynomial. Further, the same product consists of all monomials in the expansion of the determinant of $\mybar\RR$ that contain elements from the $2K$ diagonal blocks $\{ \mybar{p}_1^{(k,k)} \mybar\Q_1 \}$ and $\{ \mybar{p}_2^{(k,k+K)} \mybar\Q_2 \}$ of the top-left $2K^2 \times 2K^2$ quarter of $\mybar\RR$, and no other elements from the left half of the same matrix. 

It follows that no monomials in the first part of the right-hand side of \eqref{eq:detRbar} can be canceled out analytically by any monomials in the second part, since all monomials in the second part contain at least one element from outside the diagonal blocks. Hence, $\det (\mybar\RR) = p(\PP,\Q) \not\equiv 0$, and there exist matrices $\PP^*$ and $\Q^*$ such that $p(\PP^*, \Q^*)\neq 0$.

Next, note that with $M=N=2K=2J$, we can explicitly express $\mybar\Q_t = -\G_{t,2} \G_{t,1}^{-1}$, where $\G_{t,1}$ and $\G_{t,2}$ are the upper and lower $K \times K$ blocks of $\G_t$, respectively. Also, $\mybar\PP_t = \mybar\SSSS_t^T \V_t^T$, where $\mybar\SSSS_t^T$ contains the first $K$ rows of $\SSSS_t^T$. Define
\begin{align}
    \PP & =
    \underbrace{
    \begin{bmatrix}
        \mybar\SSSS_1^T & & &
        \\
        & \mybar\SSSS_2^T & &
        \\
        & & \mybar\SSSS_3^T &
        \\
        & & & \mybar\SSSS_4^T
    \end{bmatrix}}_{\triangleq \SSSS}
    \underbrace{
    \begin{bmatrix}
        \V_1^T
        \\
        \V_2^T
        \\
        \V_3^T
        \\
        \V_4^T
    \end{bmatrix}}_{\triangleq \V}
    \\
    \Q & = -
    \underbrace{
    \begin{bmatrix}
        \G_{1,2}^T
        \\
        \G_{2,2}^T
        \\
        \G_{3,2}^T 
        \\
        \G_{4,2}^T
    \end{bmatrix}^T}_{\triangleq \G_2}
    \underbrace{
    \begin{bmatrix}
        \G_{1,1}^{-1} & & &
        \\
        & \G_{2,1}^{-1} & &
        \\
        & & \G_{3,1}^{-1} &
        \\
        & & & \G_{4,1}^{-1}
    \end{bmatrix}}_{\triangleq \G_1^{-1}};
\end{align}
thus, for the matrices $\PP^*, \Q^*$ introduced above there exist unique matrices $\V^* = \SSSS^{-1} \PP^*$ and $\G_2^* = - \Q^* \G_1$. Let $q(\cdot)$ be the polynomial defined by a linear transformation of $p(\cdot)$, such that $q(\V, \G_{2}) = p(\PP, \Q)$. Then,
\begin{align}
    \det(\mybar\RR) & = p(\PP, \Q) = p(\SSSS \V, -\G_2 \G_1^{-1})
    \\
    & = q(\V, \G_2) \not\equiv 0,
\end{align}
since $q(\V^*, \G_2^*) = p(\PP^*, \Q^*) \neq 0$, and it follows that the polynomial $q(\V, \G_2) \not\equiv 0$. Further, $q(\V, \G_2)$ is a polynomial in the independent continuous random variables (with a support on the entire complex plane) of $\V$ and $\G_2$, so the set of realizations for which it equals zero (e.g., if $\V^*$ and $\G_2^*$ happen to be zero matrices) has measure zero. Hence, with probability one, $R = \rank(\mybar\RR) = 4K^2$ and  $\A$ is full rank with probability one if and only if $T \geq 4$.

\renewcommand{\thesectiondis}[2]{\Alph{section}:}
\section{Proof of Theorem \ref{theorem:T=1}}
\label{appendix:SC_T1}

By following the same steps as in the proof of Theorem \ref{theorem:dimNA} (with $T=1$) we get that $\rank(\A) = K \tau + \rank(\RR)$, where the matrix $\RR$ in \eqref{eq:RR} is now given by $\RR = \SSSS_1^T \V_1^T \otimes \Q_1$. Thus, $\rank(\RR) = \rank(\SSSS_1^T \V_1^T) \rank(\Q_1) = J(M-K)$, and it follows that $\rank(\A) = K\tau + (M-K)J$.

Note that this theorem extends  \cite[Th. 1]{Andersson2024ICASSP}, which assumed $J=1$, to a general value of $J$. 

\renewcommand{\thesectiondis}[2]{\Alph{section}:}
\section{Proof of Proposition \ref{prop:covxt}}
\label{appendix:covxt}
The matrix $\A^H\A$ can be partitioned as
\begin{align}
    \A^H\A =
    \begin{bmatrix}
        \B & \CC
        \\
        \CC^H & \D
    \end{bmatrix},
\end{align}
where, with $\PP_t = \SSSS_t^T \V_t^T$, we have
\begin{align}
    \nonumber
    \B & \triangleq
    \begin{bmatrix}
        \II_{\tau}\otimes \G_1^H \G_1 & & \0
        \\
        & \ddots &
        \\
        \0 & & \II_{\tau}\otimes \G_T^H \G_T
    \end{bmatrix},
    \\
    \label{eq:blockdef}
    \CC & \triangleq
    \begin{bmatrix}
        \PP_1 \otimes \G_1^H
        \\
        \vdots
        \\
        \PP_T \otimes \G_T^H
    \end{bmatrix},\
    \D \triangleq \sum_{t\in \TT} (\PP_t^H \PP_t \otimes \II_M).
\end{align}
Then, by using the inversion lemma for block-partitioned matrices \cite[Sec. 0.7.3]{horn2012matrix}, we get
\begin{align}
    \nonumber
    & \Cov(\hat\x) = [(\A^H \A)^{-1}]_{1:TK\tau,1:TK\tau}
    \\
    & = \B^{-1} + \B^{-1} \CC (\D - \CC^H \B^{-1} \CC)^{-1} \CC^H \B^{-1}.
\end{align}
Define $\check\E \triangleq \B^{-1} \CC (\D - \CC^H \B^{-1} \CC)^{-1} \CC^H \B^{-1}$ and partition this matrix into $T \times T$ blocks, each of size $K\tau \times K\tau$:
\begin{align}
    \check\E =
    \begin{bmatrix}
        \check\E_{1,1} & \dots & \check\E_{1,T}
        \\
        \vdots & \ddots & \vdots
        \\
        \check\E_{T,1} & \dots & \check\E_{T,T}
    \end{bmatrix}.
\end{align}
For any $i,j \in \{ 1,\dots,T\}$, these blocks are given by
\begin{align}
    \check\E_{i,j} \triangleq (\PP_i \otimes (\G_i^H \G_i)^{-1} \G_i^H) \E (\PP_j^H \otimes \G_j (\G_j^H \G_j)^{-1}),
\end{align}
where
\begin{align}
    \nonumber
    \E & \triangleq (\D - \CC^H \B^{-1} \CC)^{-1}
    \\
    & = \left(\sum_{t \in \TT} \PP_t^H \PP_t \otimes (\II_M - \G_t (\G_t^H \G_t)^{-1} \G_t^H) \right)^{-1}.
\end{align}
By extracting the $t$th $K\tau \times K\tau$ diagonal block of $\Cov(\hat\x)$ and inserting $\PP_t = \SSSS_t^T \V_t^T$ we get \eqref{Covxt} which completes the proof. 

\renewcommand{\thesectiondis}[2]{\Alph{section}:}
\section{Proof of Theorem \ref{theorem:dimNAbar}}
\label{appendix:dimNAbar}
To show that $\mybar\A \in \C^{M\tau \times (K(\tau-J) + MJ)}$ is full-rank we must find $K(\tau-J) + MJ$ linearly independent columns. Partition $\mybar\A$ as $\mybar\A = \begin{bmatrix} \mybar\A_1 & \mybar\A_2 \end{bmatrix}$, where
\begin{align}
    \mybar\A_1 =
    \begin{bmatrix}
        \0_{MJ \times K(\tau - J)}
        \\
        \II_{\tau-J} \otimes \G
    \end{bmatrix}
    ,\
    \mybar\A_2 = \SSSS^T \otimes \II_{M}.
\end{align}
From Assumption \ref{assum:VS} the left-most block must be full-rank, i.e. $\rank(\mybar\A_1) = K(\tau-J)$. It remains to show that $MJ$ additional linearly independent columns can be found in $\mybar\A_2$. Since the first $MJ$ rows of $\mybar\A_1$ are zero vectors, it is sufficient to show that the corresponding part of $\mybar\A_2$, i.e. its upper-most $MJ \times MJ$ block, is full-rank. But this is clear since $\SSSS^T$ has i.i.d. elements and $\rank(\mybar\A_2) = \rank(\SSSS^T \otimes \II_{M}) = MJ$.

\renewcommand{\thesectiondis}[2]{\Alph{section}:}
\section{Proof of Theorem \ref{theorem:eqMethods}}
\label{appendix:eqMethods}

The proof follows the same lines as \cite[Th. 2]{Andersson2024ICASSP}, which considered the special case of $J=1$. With $\Y = \begin{bmatrix} \Y_{1:J} & \mybar\Y \end{bmatrix}$ and $\SSSS = \begin{bmatrix} \SSSS_{1:J} & \mybar\SSSS \end{bmatrix}$, the least-squares problem solved in Algorithm 1 (joint estimation) is, on matrix form,
\begin{align}
    \label{sc2:ls_matrix}
    \nonumber
    & \min_{ \mybar\X, \F} \| \Y - \G
    \begin{bmatrix}
        \0_{K \times J} & \mybar\X
    \end{bmatrix}
    - \F \SSSS \|^2
    \\
    & = \min_{ \mybar\X, \F} \{ \| \Y_{1:J} - \F \SSSS_{1:J} \|^2 + \| \mybar\Y - \G \mybar\X - \F \mybar\SSSS \|^2 \}.
\end{align}
From \eqref{sc2:ls_matrix}, the solution to $\mybar\X$ for a given matrix $\F$ is
\begin{align}
    \label{eq:hatXbar}
    \hat{\mybar\X} = (\G^H \G)^{-1} \G^H (\mybar\Y - \F \mybar\SSSS).
\end{align}
By inserting \eqref{eq:hatXbar} into \eqref{sc2:ls_matrix}, we obtain the remaining problem to minimize over $\F$:
\begin{align}
    \label{sc2:ls_F}
    \nonumber
    & \min_{\F} \{ \| \Y_{1:J} - \F \SSSS_{1:J} \|^2 
    \\
    \nonumber
    & \qquad + \| \mybar\Y - \G (\G^H \G)^{-1} \G^H (\mybar\Y - \F \mybar\SSSS) - \F \mybar\SSSS  \|^2 \}
    \\
    & = \min_{\F} \{ \| \Y_{1:J} - \F \SSSS_{1:J} \|^2 + \| \Pi_{\G}^{\perp} (\mybar\Y - \F \mybar\SSSS) \|^2 \},
\end{align}
where $\Pi_{\G}^{\perp} \triangleq \II_M - \G (\G^H \G)^{-1} \G^H$ is the projection matrix onto the orthogonal complement of the column space of $\G$. Next, decompose $\Y_{1:J} = \Y_{1:J}' + \Y_{1:J}''$ and $\F = \F' + \F''$ such that $\Y_{1:J}',\F' \in \mathrm{sp}(\G)$ and $\Y_{1:J}'',\F'' \in \mathrm{sp}(\G)^{\perp}$, where $\mathrm{sp}(\G)$ respectively $\mathrm{sp}(\G)^{\perp}$ denote the span of the columns of $\G$ and its orthogonal complement. Now, we write \eqref{sc2:ls_F} as
\begin{align}
    \label{sc2:ls_F'F''}
    \nonumber
    & \min_{\F} \{ \| \Y_{1:J}' + \Y_{1:J}'' - \F' \SSSS_{1:J} - \F'' \SSSS_{1:J} \|^2 
    \\
    \nonumber
    & \qquad + \| \Pi_{\G}^{\perp} (\mybar\Y - \F' \mybar\SSSS - \F'' \mybar\SSSS) \|^2 \}
    \\
    \nonumber
    & = \min_{\F} \{ \| \Y_{1:J}' - \F' \SSSS_{1:J} \|^2 + \| \Y_{1:J}'' - \F'' \SSSS_{1:J} \|^2 
    \\
    & \quad \qquad + \| \Pi_{\G}^{\perp} (\mybar\Y - \F'' \mybar\SSSS) \|^2 \}.
\end{align}
Let the solution to \eqref{sc2:ls_F'F''} be $\hat\F_\star = \hat\F_\star'  + \hat\F_\star''$, with $\hat\F_\star' \in \mathrm{sp}(\G)$ and $\hat\F_\star'' \in \mathrm{sp}(\G)^{\perp}$. Solving \eqref{sc2:ls_F'F''} with respect to $\F'$ gives $\hat\F_\star' = \Y_{1:J}' \SSSS_{1:J}^H (\SSSS_{1:J} \SSSS_{1:J}^H)^{-1}$. Thus, we get
\begin{align}
    \label{eq:solXbar}
    \nonumber
    \hat{\mybar\X} & = (\G^H \G)^{-1} \G^H (\mybar\Y - \hat\F_\star \mybar\SSSS) = (\G^H \G)^{-1} \G^H (\mybar\Y - \hat\F_\star' \mybar\SSSS)
    \\
    \nonumber
    & = (\G^H \G)^{-1} \G^H (\mybar\Y - \Y_{1:J}' \SSSS_{1:J}^H (\SSSS_{1:J} \SSSS_{1:J}^H)^{-1} \mybar\SSSS)
    \\
    \nonumber
    & = (\G^H \G)^{-1} \G^H (\mybar\Y - \Y_{1:J} \SSSS_{1:J}^H (\SSSS_{1:J} \SSSS_{1:J}^H)^{-1} \mybar\SSSS)
    \\
    & = (\G^H \G)^{-1} \G^H (\mybar\Y - \hat\F \mybar\SSSS),
\end{align}
where in the second and fourth equalities, respectively, we used that $\hat\F_\star'', \Y_{1:J}'' \in \mathrm{sp}(\G)^{\perp}$, and in the final step we recognized $\Y_{1:J} \SSSS_{1:J}^H (\SSSS_{1:J} \SSSS_{1:J}^H)^{-1}$ as the Algorithm~2 channel estimate $\hat\F$ in \eqref{eq:hatF}. Hence \eqref{eq:solXbar} is exactly the solution in \eqref{XBarHat_C2} with Algorithm 2, and the two solutions are equal.

\renewcommand{\thesectiondis}[2]{\Alph{section}:}
\renewcommand{\thesectiondis}[2]{\Alph{section}:}

\section{Algorithm and Complexity Analysis for the Full-Rank Case}
\label{appendix:complexity}

The unknown vector $\z = [\x_1^T, \dots, \x_T^T, \h^T]^T$ is partitioned to match the bordered block-diagonal structure of $\A$ in \eqref{A}.
As established in the proof of Proposition~\ref{prop:covxt}, the Gram matrix $\A^H \A$ then admits the $2 \times 2$ block partition with $\B$, $\CC$, and $\D$ defined in \eqref{eq:blockdef}, and the Schur complement of $\B$, i.e. $\E^{-1} = \D - \CC^H \B^{-1} \CC$, has the closed-form inverse $\E$ in \eqref{E}.

Rather than forming $(\A^H \A)^{-1}$, we solve the normal equations
$\A^H \A \hat\z = \A^H \y$ in block form,
\begin{align}
    \label{eq:normal_x}
    \B \hat\x + \CC \hat\h & = (\A^H \y)_{\x}, \\
    \label{eq:normal_h}
    \CC^H \hat\x + \D \hat\h & = \sum_{t \in \TT} (\PP_t^H \otimes \II_M) \y_t,
\end{align}
where $(\A^H \y)_{\x}$ is the $\x$-block of $\A^H \y$, with $t$th sub-block
$(\II_\tau \otimes \G_t^H)\y_t$. Eliminating $\hat\x = \B^{-1}((\A^H \y)_{\x}
- \CC \hat\h)$ from \eqref{eq:normal_x} and substituting into
\eqref{eq:normal_h} yields the reduced system whose left-hand side is the
Schur complement $\E^{-1}$; evaluating the right-hand side gives
\begin{align}
    \label{eq:linsys_h}
    \E^{-1} \hat\h = \mathbf{r} \triangleq \sum_{t \in \TT}
    \big( \PP_t^H \otimes (\II_M - \G_t (\G_t^H \G_t)^{-1} \G_t^H) \big) \y_t.
\end{align}
Once $\hat\h$ is known, back-substitution in the $t$th sub-block of
\eqref{eq:normal_x} gives the $T$ decoupled per-interval estimates
\begin{align}
    \label{eq:hatxt_alg}
    \hat\x_t = \big( \II_\tau \otimes (\G_t^H \G_t)^{-1} \G_t^H \big)
    \big( \y_t - (\PP_t \otimes \II_M) \hat\h \big).
\end{align}
Equation \eqref{eq:linsys_h} is a single $MN \times MN$ system for $\hat\h$ in
which the $T$ coherence intervals enter only through the sum $\mathbf{r}$; this
decoupling is what makes the cost linear in $T$.

We solve \eqref{eq:linsys_h} without ever forming $\E$: since $\E^{-1}$ is
available in closed form, we Cholesky-factorize it once as $\E^{-1} =
\mathbf{L} \mathbf{L}^H$, with $\mathbf{L}$ lower-triangular. Introducing the
auxiliary vector $\mathbf{c} \triangleq \mathbf{L}^H \hat\h$, the system
$\mathbf{L} \mathbf{L}^H \hat\h = \mathbf{r}$ decouples into two triangular
sub-systems, $\mathbf{L} \mathbf{c} = \mathbf{r}$ and $\mathbf{L}^H \hat\h =
\mathbf{c}$. We first recover $\mathbf{c}$ by \emph{forward substitution}:
since $\mathbf{L}$ is lower-triangular, its $i$th equation reads
$L_{ii} c_i = r_i - \sum_{j<i} L_{ij} c_j$, so sweeping $i$ from $1$ to $MN$
yields $\mathbf{c}$ at cost $\OO((MN)^2)$. We then recover $\hat\h$ by
\emph{backward substitution}: since $\mathbf{L}^H$ is upper-triangular, its
$i$th equation reads $L^*_{ii} \hat h_i = c_i - \sum_{j>i} L^*_{ji} \hat h_j$,
so sweeping $i$ from $MN$ down to $1$ yields $\hat\h$ at the same cost. The
factorization itself costs $\OO((MN)^3)$ \cite[Sec.~4.2.5]{golub2013matrix}
and dominates the two triangular solves. Since $\E^{-1}$ has size $MN \times MN$, independent of $T$, this factorization cost is
incurred only once, no matter how large $T$ grows.

The remaining work scales linearly with $T$, since it consists of sums over
the $T$ intervals: assembling $\E^{-1}$ and the right-hand side $\mathbf{r}$
of \eqref{eq:linsys_h}, and applying \eqref{eq:hatxt_alg} once per interval.
Per interval, assembling $\E^{-1}$ costs $\OO(N^2\tau)$ for $\PP_t^H \PP_t$,
$\OO(M^2 K)$ for the projector $\II_M - \G_t (\G_t^H \G_t)^{-1} \G_t^H$, and
$\OO((MN)^2)$ for the Kronecker product. The terms of $\mathbf{r}$ and of
\eqref{eq:hatxt_alg} are evaluated without materializing any Kronecker matrix,
via the identity \eqref{veckron};
for instance, reshaping
$\y_t$ into $\mathbf{Y}_t \in \C^{M\times\tau}$, the $t$th term of $\mathbf{r}$
equals $\mathrm{vec}\big((\II_M - \G_t (\G_t^H \G_t)^{-1} \G_t^H) \mathbf{Y}_t
\PP_t^*\big)$, i.e. two matrix products at $\OO(M^2\tau + MN\tau)$. The
per-interval cost is thus $\OO((MN)^2 + (M+N)^2\tau)$, absorbing lower-order
terms in $K \leq M,N$.

Summing over the $T$ intervals, the dominant total cost is
\begin{align}
    \label{eq:complexity_total}
    \OO\Big( T \cdot \big( (MN)^2 + (M+N)^2 \tau \big) + (MN)^3 \Big),
\end{align}
i.e.\ linear in $T$ plus a one-time $\OO((MN)^3)$ cost for the AP-AP channel
system. By contrast, since $\z$ has $TK\tau + MN$ entries, $\A^H \A$ is a dense
$(TK\tau + MN) \times (TK\tau + MN)$ matrix, so na\"ively forming and inverting
it by standard dense linear algebra costs $\OO((TK\tau + MN)^3)$, which scales
as $T^3$.

\bibliographystyle{IEEEtran}
\bibliography{IEEEabrv,refs}

\begin{thebibliography}{10}
\providecommand{\url}[1]{#1}
\csname url@samestyle\endcsname
\providecommand{\newblock}{\relax}
\providecommand{\bibinfo}[2]{#2}
\providecommand{\BIBentrySTDinterwordspacing}{\spaceskip=0pt\relax}
\providecommand{\BIBentryALTinterwordstretchfactor}{4}
\providecommand{\BIBentryALTinterwordspacing}{\spaceskip=\fontdimen2\font plus
\BIBentryALTinterwordstretchfactor\fontdimen3\font minus
  \fontdimen4\font\relax}
\providecommand{\BIBforeignlanguage}[2]{{%
\expandafter\ifx\csname l@#1\endcsname\relax
\typeout{** WARNING: IEEEtran.bst: No hyphenation pattern has been}%
\typeout{** loaded for the language `#1'. Using the pattern for}%
\typeout{** the default language instead.}%
\else
\language=\csname l@#1\endcsname
\fi
#2}}
\providecommand{\BIBdecl}{\relax}
\BIBdecl

\bibitem{Andersson2024ICASSP}
M.~Andersson, T.~T. Vu, P.~Frenger, and E.~G. Larsson, ``Uplink symbol
  detection in dynamic {TDD} {MIMO} systems with {AP-AP} interference,'' in
  \emph{Proc. IEEE Int. Conf. Acoust. Speech. Signal Process. (ICASSP)}, 2024,
  pp. 9236--9240.

\bibitem{ngo16}
T.~L. Marzetta, E.~G. Larsson, H.~Yang, and H.~Q. Ngo, \emph{Fundamentals of
  Massive {MIMO}}.\hskip 1em plus 0.5em minus 0.4em\relax Cambridge University
  Press, 2016.

\bibitem{emil20TWC}
E.~{Bj\"{o}rnson} and L.~{Sanguinetti}, ``Making cell-free massive {MIMO}
  competitive with {MMSE} processing and centralized implementation,''
  \emph{IEEE Trans. Wireless Commun.}, vol.~19, no.~1, pp. 77--90, Jan. 2020.

\bibitem{ammar22CST}
H.~A. Ammar, R.~Adve, S.~Shahbazpanahi, G.~Boudreau, and K.~V. Srinivas,
  ``User-centric cell-free massive {MIMO} networks: A survey of opportunities,
  challenges and solutions,'' \emph{IEEE Commun. Surveys Tuts.}, vol.~24,
  no.~1, pp. 611--652, 2021.

\bibitem{Andersson2023Asilomar}
M.~Andersson, T.~T. Vu, P.~Frenger, and E.~G. Larsson, ``Joint optimization of
  switching point and power control in dynamic {TDD} cell-free massive
  {MIMO},'' in \emph{Proc. IEEE Asilomar Conf. Signal. Syst. and Comput.
  (ACSSC)}, 2023, pp. 988--992.

\bibitem{KimCST}
H.~Kim, J.~Kim, and D.~Hong, ``Dynamic {TDD} systems for {5G} and beyond: A
  survey of cross-link interference mitigation,'' \emph{IEEE Commun. Surveys
  Tut.}, vol.~22, no.~4, pp. 2315--2348, 2020.

\bibitem{Mohammadi2023JSAC}
M.~Mohammadi, T.~T. Vu, H.~Q. Ngo, and M.~Matthaiou, ``Network-assisted
  full-duplex cell-free massive {MIMO}: Spectral and energy efficiencies,''
  \emph{IEEE J. Sel. Areas Commun.}, vol.~41, no.~9, pp. 2833--2851, 2023.

\bibitem{Chowdhury2024TCOM}
A.~Chowdhury and C.~R. Murthy, ``Half-duplex {APs} with dynamic {TDD} versus
  full-duplex {APs} in cell-free systems,'' \emph{IEEE Trans. Commun.},
  vol.~72, no.~7, pp. 3856--3872, 2024.

\bibitem{5GNR_crosslink}
3GPP, ``Discussion on duplexing flexibility and cross-link interference
  mitigation schemes,'' \emph{3{GPP TSG RAN WG}1 {M}eeting \#88, R1-1701616},
  2017.

\bibitem{3GPP_config}
------, ``Further enhancements to {LTE TDD} for {DL-UL} interference management
  and traffic adaptation,'' \emph{3{GPP TSG RAN} {M}eeting \#51, RP-110450},
  2011.

\bibitem{Zhang2015CM}
Z.~Zhang, X.~Chai, K.~Long, A.~V. Vasilakos, and L.~Hanzo, ``Full duplex
  techniques for 5{G} networks: self-interference cancellation, protocol
  design, and relay selection,'' \emph{IEEE Commun. Mag.}, vol.~53, no.~5, pp.
  128--137, 2015.

\bibitem{Smida2024PI}
B.~Smida, R.~Wichman, K.~E. Kolodziej, H.~A. Suraweera, T.~Riihonen, and
  A.~Sabharwal, ``In-band full-duplex: The physical layer,'' \emph{Proc. of the
  IEEE}, vol. 112, no.~5, pp. 433--462, 2024.

\bibitem{Kim2024PI}
Y.~Kim, H.-J. Moon, H.~Yoo, B.~Kim, K.-K. Wong, and C.-B. Chae, ``A
  state-of-the-art survey on full-duplex network design,'' \emph{Proc. of the
  IEEE}, vol. 112, no.~5, pp. 463--486, 2024.

\bibitem{Mohammadi2025TCOM}
M.~Mohammadi, Z.~Mobini, H.~Quoc~Ngo, and M.~Matthaiou, ``Ten years of research
  advances in full-duplex massive {MIMO},'' \emph{IEEE Trans. Commun.},
  vol.~73, no.~3, pp. 1756--1786, 2025.

\bibitem{chowdhury2021TC}
A.~Chowdhury, R.~Chopra, and C.~R. Murthy, ``Can dynamic {TDD} enabled
  half-duplex cell-free massive {MIMO} outperform full-duplex cellular massive
  {MIMO}?'' \emph{IEEE Trans. Commun.}, vol.~70, no.~7, pp. 4867--4883, 2022.

\bibitem{Huang2019SJ}
Y.~Huang, B.~Jalaian, S.~Russell, and H.~Samani, ``Reaping the benefits of
  dynamic {TDD} in massive {MIMO},'' \emph{IEEE Systems J.}, vol.~13, no.~1,
  pp. 117--124, 2019.

\bibitem{Zhu2021CL}
Y.~Zhu, J.~Li, P.~Zhu, H.~Wu, D.~Wang, and X.~You, ``Optimization of duplex
  mode selection for network-assisted full-duplex cell-free massive {MIMO}
  systems,'' \emph{IEEE Commun. Lett.}, vol.~25, no.~11, pp. 3649--3653, 2021.

\bibitem{Razlighi2021TWC}
M.~M. Razlighi, N.~Zlatanov, S.~R. Pokhrel, and P.~Popovski, ``Optimal
  centralized dynamic-time-division-duplex,'' \emph{IEEE Trans. Wireless
  Commun.}, vol.~20, no.~1, pp. 28--39, 2021.

\bibitem{fukue22A}
S.~Fukue, H.~Iimori, G.~T.~F. De~Abreu, and K.~Ishibashi, ``Joint access
  configuration and beamforming for cell-free massive {MIMO} systems with
  dynamic {TDD},'' \emph{IEEE Access}, vol.~10, pp. 40\,130--40\,149, 2022.

\bibitem{Yoon2015CL}
C.~Yoon and D.-H. Cho, ``Energy efficient beamforming and power allocation in
  dynamic {TDD} based {C-RAN} system,'' \emph{IEEE Commun. Lett.}, vol.~19,
  no.~10, pp. 1806--1809, 2015.

\bibitem{deOlivindo2018WCL}
E.~de~Olivindo~Cavalcante, G.~Fodor, Y.~C.~B. Silva, and W.~C. Freitas,
  ``Distributed beamforming in dynamic {TDD} {MIMO} networks with {BS} to {BS}
  interference constraints,'' \emph{IEEE Wireless Commun. Lett.}, vol.~7,
  no.~5, pp. 788--791, 2018.

\bibitem{Jayasinghe2018TSP}
P.~Jayasinghe, A.~Tölli, J.~Kaleva, and M.~Latva-aho, ``Bi-directional
  beamformer training for dynamic {TDD} networks,'' \emph{IEEE Trans. Signal
  Process.}, vol.~66, no.~23, pp. 6252--6267, 2018.

\bibitem{daSilva2021WC}
J.~M.~B. da~Silva, G.~Wikström, R.~K. Mungara, and C.~Fischione, ``Full duplex
  and dynamic {TDD}: Pushing the limits of spectrum reuse in multi-cell
  communications,'' \emph{IEEE Wireless Commun.}, vol.~28, no.~1, pp. 44--50,
  2021.

\bibitem{horn1991topics}
R.~A. Horn and C.~R. Johnson, \emph{Topics in Matrix Analysis}.\hskip 1em plus
  0.5em minus 0.4em\relax Cambridge, U.K.: Cambridge University Press, 1991.

\bibitem{Roth1952}
W.~E. Roth, ``The equations {$AX - YB = C$} and {$AX - XB = C$} in matrices,''
  \emph{Proceedings of the American Mathematical Society}, vol.~3, no.~3, pp.
  392--396, 1952.

\bibitem{Baksalary1979LAA}
J.~Baksalary and R.~Kala, ``The matrix equation {$AX - YB = C$},'' \emph{Linear
  Algebra Appl.}, vol.~25, pp. 41--43, 1979.

\bibitem{Ding2006SIAM}
F.~Ding and T.~Chen, ``On iterative solutions of general coupled matrix
  equations,'' \emph{SIAM J. Control Optim.}, vol.~44, no.~6, pp. 2269--2284,
  2006.

\bibitem{khatri1968}
C.~G. Khatri and C.~R. Rao, ``Solutions to some functional equations and their
  applications to characterization of probability distributions,''
  \emph{Sankhya: The Indian Journal of Statistics, Series A (1961-2002)},
  vol.~30, no.~2, pp. 167--180, 1968.

\bibitem{Sidiropoulos2000JoC}
N.~D. Sidiropoulos and R.~Bro, ``On the uniqueness of multilinear decomposition
  of {N}-way arrays,'' \emph{J. Chemom.}, vol.~14, no.~3, pp. 229--239, 2000.

\bibitem{DeLathauwer2008SIAM}
L.~De~Lathauwer, ``Decompositions of a higher-order tensor in block
  terms—part {I}: Lemmas for partitioned matrices,'' \emph{SIAM J. Matrix
  Anal. Appl.}, vol.~30, no.~3, pp. 1022--1032, 2008.

\bibitem{tr36932}
3GPP, ``Scenarios and requirements for small cell enhancements for {E-UTRA} and
  {E-UTRAN},'' \emph{3{GPP TR}~36.932}, 2025.

\bibitem{emil17}
E.~Bj\"{o}rnson, J.~Hoydis, and L.~Sanguinetti, ``Massive {MIMO} networks:
  Spectral, energy, and hardware efficiency,'' \emph{Found. Trends. Signal
  Process.}, vol.~11, no. 3-4, pp. 154--655, 2017.

\bibitem{bai2024JSTSP}
J.~Bai and E.~G. Larsson, ``A unified activity detection framework for massive
  access: Beyond the block-fading paradigm,'' \emph{IEEE J. Sel. Topics Signal
  Process.}, vol.~18, no.~7, pp. 1366--1380, 2024.

\bibitem{horn2012matrix}
R.~A. Horn and C.~R. Johnson, \emph{Matrix analysis}.\hskip 1em plus 0.5em
  minus 0.4em\relax Cambridge, U.K.: Cambridge university press, 2012.

\bibitem{lancaster1985theory}
P.~Lancaster and M.~Tismenetsky, \emph{The Theory of Matrices}.\hskip 1em plus
  0.5em minus 0.4em\relax New York, NY, USA: Academic Press, 1985.

\bibitem{golub2013matrix}
G.~H. Golub and C.~F. Van~Loan, \emph{Matrix Computations}, 4th~ed.\hskip 1em
  plus 0.5em minus 0.4em\relax Baltimore, MD, USA: Johns Hopkins University
  Press, 2013.

\end{thebibliography}

\end{document}